\documentclass[a4paper,UKenglish,cleveref, autoref, thm-restate]{lipics-v2019}
\title{Wreath/cascade products and related decomposition results for the concurrent setting of Mazurkiewicz traces (extended version)} %TODO Please add

\titlerunning{Wreath products in the concurrent setting (extended version)} %TODO optional, please use if title is longer than one line

\author{Bharat Adsul}{IIT Bombay, India}{adsul@cse.iitb.ac.in}{}{}%TODO mandatory, please use full name; only 1 author per \author macro; first two parameters are mandatory, other parameters can be empty. Please provide at least the name of the affiliation and the country. The full address is optional

\author{Paul Gastin}{LSV, ENS Paris-Saclay, CNRS, Universit\'e Paris-Saclay,France}{paul.gastin@ens-paris-saclay.fr}{https://orcid.org/0000-0002-1313-7722}{Supported by IRL ReLaX}
\author{Saptarshi Sarkar}{IIT Bombay, India}{sapta@cse.iitb.ac.in}{}{}
\author{Pascal Weil}{Univ. Bordeaux, LaBRI, CNRS UMR 5800, F-33400 Talence, France\\CNRS, ReLaX, IRL 2000, Siruseri, India}{pascal.weil@labri.fr}{}{Partially supported by the DeLTA project (ANR-16-CE40- 0007)}

\authorrunning{B.\,Adsul, P.\,Gastin, S.\,Sarkar and P.\,Weil} %TODO mandatory. First: Use abbreviated first/middle names. Second (only in severe cases): Use first author plus 'et al.'

\Copyright{Bharat Adsul, Paul Gastin, Saptarshi Sarkar, and Pascal Weil} %TODO mandatory, please use full first names. LIPIcs license is "CC-BY";  http://creativecommons.org/licenses/by/3.0/


 %TODO mandatory; please add comma-separated list of keywords

\nolinenumbers %uncomment to disable line numbering

\hideLIPIcs  %uncomment to remove references to LIPIcs series (logo, DOI, ...), e.g. when preparing a pre-final version to be uploaded to arXiv or another public repository

\EventEditors{John Q. Open and Joan R. Access}
\EventNoEds{2}
\EventLongTitle{42nd Conference on Very Important Topics (CVIT 2016)}
\EventShortTitle{CVIT 2016}
\EventAcronym{CVIT}
\EventYear{2016}
\EventDate{December 24--27, 2016}
\EventLocation{Little Whinging, United Kingdom}
\EventLogo{}
\SeriesVolume{42}
\ArticleNo{23}
\usepackage{tikz}
\usetikzlibrary{calc}
\usetikzlibrary{backgrounds}
\usetikzlibrary{arrows,automata}

\graphicspath{{figures/}}
\newtheorem{question}{Question}

\newcommand{\alphabet}{\Sigma}
\newcommand{\tralphabet}{\alphabet^{\|_S}}
\newcommand{\tralphabetgossip}{\alphabet^{\|_{S^\gossip}}}
\newcommand{\gtralphabet}{\alphabet^{\gs}}
\newcommand{\dtralphabet}{\widetilde{\tralphabet}}
\newcommand{\dgtralphabet}{\widetilde{\gtralphabet}}
\newcommand{\eat}[1]{}
\newcommand{\pset}{\mathcal{P}}
\newcommand{\dalphabet}{\widetilde{\alphabet}}
\newcommand{\loc}{\mathrm{loc}}
\newcommand{\ind}{I}
\newcommand{\dep}{D}
\newcommand{\labf}{\lambda}
\newcommand{\isucc}{\lessdot}
\newcommand{\dcset}[1]{{\downarrow}#1}
\newcommand{\traces}{TR(\dalphabet)}
\newcommand{\tracesdtr}{TR(\dtralphabet\!)}
\newcommand{\tracesdgtr}{TR(\dgtralphabet\!)}
\newcommand{\configs}[1]{\mathcal{C}_{#1}}
\newcommand{\view}[2]{{\downarrow}^{#1}(#2)}
\newcommand{\pview}[2]{{\downarrow}_{\mathrm{prev}}^{#1}(#2)}
\newcommand{\ls}{S}
\newcommand{\gs}{\ls_{\pset}}
\newcommand{\asa}{A}
\newcommand{\lt}{\delta}
\newcommand{\gt}{\Delta}
\newcommand{\et}[1]{\Rightarrow_{#1}}
\newcommand{\etl}[2]{\overset{#2}{\Rightarrow}_{#1}}
\newcommand{\at}[1]{\rightarrow_{#1}}
\newcommand{\atl}[2]{\xrightarrow{#2}_{#1}}
\newcommand{\run}{\rho}
\newcommand{\gf}{\ls_{\text{fin}}}

\newcommand{\tmi}{(X,M)}
\newcommand{\tmii}{(Y,N)}
\newcommand{\divides}{\prec}
\newcommand{\atma}{(\{\ls_i\},M)}

\newcommand{\wt}{\widetilde}
\newcommand{\gossip}{\mathcal{G}}
\newcommand{\latest}{\mathrm{latest}}
\newcommand{\globalstate}{\mathrm{globalstate}}
\newcommand{\gc}{\circ_g}
\newcommand{\lc}{\circ_{\ell}}

\newcommand{\bbb}{B}

\let\phi\varphi
\let\ab\allowbreak
\let\subset\subseteq

\newcommand{\s}{\Sigma}

\newcommand{\aaa}{A}

\newcommand{\trans}{\mathcal{F}}

\newcommand{\wpis}{X \times Y}

\newcommand{\Y}{\mathop{\mathsf{Y}\vphantom{a}}\nolimits}
\newcommand{\Si}{\mathbin{\mathsf{S}}}
\newcommand{\loctl}{\mathsf{LocTL}[\Y_i,\Si_i]}
\newcommand{\sprtl}{\mathsf{LocTL}[\Si_i]}
\newcommand{\Oo}{\mathop{\mathsf{O}\vphantom{a}}\nolimits}

\begin{document}

\maketitle

\begin{abstract}
We develop a new algebraic framework to reason about languages of Mazurkiewicz
traces. This framework supports true concurrency and provides a non-trivial generalization
of the wreath product operation to the trace setting.
A novel local wreath product principle has been established. 
The new framework is crucially used to propose a decomposition result for 
recognizable trace languages, which is an analogue of the Krohn-Rhodes theorem. 
We prove this decomposition result in the special case of acyclic architectures
and apply it to extend Kamp's theorem to this setting.
We also introduce and analyze distributed automata-theoretic operations called local and global cascade products.
Finally, we show that aperiodic trace languages can be characterized using global cascade
products of localized and distributed two-state reset automata.
\end{abstract}

\section{Introduction}
Transformation monoids provide an abstraction of transition systems.
One of the key tools in their analysis is the notion of 
wreath product \cite{Eilenberg1976AutomataLA,str_cirBook,MPRI-notes} which, when translated to the language of finite state automata, 
corresponds to the cascade product. In the 
cascade product of automata $A$ and $B$, with $A$ `followed by' $B$, 
the automaton $A$ runs on the input sequence, while the automaton $B$ runs
on the input sequence as well as the state sequence produced by
the automaton $A$. The wreath product principle (see \cite{str_cirBook,MPRI-notes,infinitewords}) is a key result
which relates a language accepted by a cascade/wreath product
to languages accepted by the individual automata. 

In this work, we are interested in generalizing the wreath product operation
from the sequential setting to the concurrent setting involving multiple 
processes. Towards this, 
we work with Mazurkiewicz traces (or simply traces) \cite{Mazurkiewicz_1977, traces-book} which are 
well established as models of true concurrency, and asynchronous automata \cite{zielonka1987notes}
which are natural distributed finite state devices working on traces.
A trace represents a concurrent behaviour as a labelled partial order which 
faithfully captures the distribution of events across processes, and
causality and concurrency between them. An asynchronous automaton runs
on the input trace in a distributed fashion and respects the underlying
causality and concurrency between events. During the run, when working on an event,
only the local states of the processes participating in that event are
updated; the rest of the processes remain oblivious to the occurrence of
the event at this point.

A natural generalization of the above mentioned sequential
cascade product to asynchronous automata $A$ and $B$ is as follows: the 
asynchronous automaton $A$ runs on the input trace, 
thus assigning, for each event, a local state for every process participating
in that event. Now the asynchronous automaton $B$ runs on the input trace 
with the \emph{same} set of events which are \emph{additionally} labelled by the 
previous local states of the participating processes in $A$. 
It is easy to capture this operational semantics by another asynchronous
automaton which we call the \emph{local} cascade product of $A$ and $B$.
Such a construction is used in \cite{DBLP:conf/fsttcs/AdsulS04} 
to provide an asynchronous automata-theoretic characterization of aperiodic
trace languages.

Here we propose a new algebraic framework to deal with the issues
posed by the concurrent setting. More precisely, we introduce a new class
of transformation monoids called \emph{asynchronous transformation monoids} (in short, \emph{atm}).
These monoids make a clear distinction between local and global `states'
and allow us to reason about whether a global transformation is 
essentially induced by a particular subset of processes. 
Recall that, from a purely algebraic viewpoint, the set of all traces 
forms a free partially commutative monoid in which \emph{independent} actions
commute \cite{traces-book}. In order to recognize a trace language via an atm, we introduce the notion of an \emph{asynchronous
morphism} which exploits the locality of the underlying atm. It is rather
easy to see that asynchronous morphisms are the 
algebraic counterparts of asynchronous automata.

One of the central results of this work 
is a wreath product principle in the new algebraic framework. 
It turns out that the \emph{standard} wreath 
product operation yields an operation on asynchronous transformation monoids. 
Let $T_1$ and $T_2$ be atm's and $T_1 \wr T_2$ be the wreath product atm.
Our \emph{local} wreath product principle describes a trace language 
recognized by $T_1 \wr T_2$  in terms of a \emph{local asynchronous transducer} 
which is a natural \emph{causality and concurrency preserving} map from traces 
to traces (over an appropriately extended alphabet), and trace languages recognized by $T_1$ and $T_2$. It is a novel generalization of the
standard wreath product principle.
The work \cite{GUAIANA1998277} presents a wreath product principle for traces in the setting
of transformation monoids but it seems less significant since it uses \emph{non-trace} structures.

The importance of the standard wreath product operation is 
clearly highlighted by the fundamental Krohn-Rhodes decomposition theorem \cite{KR}
which, broadly speaking, says that any finite transformation monoid can 
be simulated by wreath products of `simple' transformation monoids. 
The wreath product
principle along with the Krohn-Rhodes theorem can be used to provide 
alternate and conceptually simpler proofs (see \cite{starfreenote-meyer,cohen1993expressive}) of several interesting classical 
results about formal 
languages of words such as Sch\"utzenberger's theorem \cite{SCHUTZENBERGER1965190}, 
McNaughton-Papert's theorem \cite{McNaughton-Papert}
and Kamp's theorem \cite{Kamp1968-KAMTLA} which together show the equivalence between
star-free, aperiodic, first-order-definable and linear-temporal-logic  definable word languages.
Motivated by these applications, we investigate an analogue of the fundamental 
Krohn-Rhodes decomposition theorem over traces.
We use the new algebraic framework to propose a simultaneous generalization
of the Krohn-Rhodes theorem (for word languages) and the Zielonka theorem (for trace 
languages). The proof of this generalization for the special case of
acyclic architectures is another significant result. As an application, we extend
Kamp's theorem: we formulate a natural local temporal logic and show that it
is expressively complete.

It turns out that asynchronous
morphisms into wreath products correspond to the aforementioned 
distributed automata-theoretic
local cascade products. 
We also introduce the \emph{global} cascade product operation
and show that it
can be realized as the local 
cascade product with the help of the ubiquitous gossip automaton from \cite{DBLP:journals/dc/MukundS97}. 

Our final major contribution concerns aperiodic trace languages and is
in the spirit of the Krohn-Rhodes theorem for the aperiodic case.
We establish that aperiodic trace languages can be characterized using global cascade
products of localized and distributed two-state reset automata.
The proof of this characterization crucially uses an expressively complete
process-based local temporal logic over traces from \cite{DBLP:journals/iandc/DiekertG06}.

The rest of the paper is organized as follows. After setting up the
preliminaries in Section~\ref{sec:prelim}, we develop the new algebraic
framework in Section~\ref{sec:alg}. In Section~\ref{sec:wpp},
we establish the local wreath product principle.
In Section~\ref{sec:dec}, we postulate a new decomposition result, and we establish it for
acyclic architectures.
We introduce and analyze local and global cascade products in Section~\ref{sec:cascade}.
The global cascade product based characterization of aperiodic trace languages appears
in Section~\ref{sec:aperiodic}. Finally, we conclude in Section~\ref{sec:conclusion}.

\section{Preliminaries}\label{sec:prelim}
\subsection{Basic Notions in Trace Theory}\label{section:traces} 
Let $\pset$ be a finite set of agents/processes. A
\emph{distributed alphabet} over $\pset$ is a family 
$\dalphabet = \{\alphabet_i\}_{i \in \pset}$. Let $\alphabet = 
\bigcup_{i \in \pset} \alphabet_i$. For $a \in \alphabet$, 
we set $\loc(a) = \{i \in \pset ~|~ a \in \alphabet_i\}$.
By $(\alphabet, \ind)$ we denote the corresponding trace 
alphabet, i.e., $\ind$ is the \emph{independence relation} 
$\{(a,b) \in \alphabet^2 ~|~ \loc(a) \cap \loc(b) 
= \emptyset \}$ induced by $\dalphabet$. The corresponding 
\emph{dependence relation} $\alphabet^2 \setminus \ind$ is
denoted by $\dep$.

A $\alphabet$-labelled poset is a structure $t = (E, \leq, 
\labf)$ where $E$ is a set, $\leq$ is a partial order on 
$E$ and $\labf \colon E \to \alphabet$ is a labelling function.
For $e, e' \in E$, define $e \isucc e'$ if and only if $e < e'$ and 
for each $e''$ with $e \leq e'' \leq e'$ either $e = e''$ 
or $e' = e''$. For $X \subseteq E$, let $\dcset{X} = 
\{y \in E ~|~ y \leq x \text{ for some } x \in X \}$. For 
$e \in E$, we abbreviate $\dcset{\{e\}}$ by simply $\dcset e$. 

A \emph{trace} over $\dalphabet$ is a finite 
$\alphabet$-labelled poset $t = (E, \leq, \labf)$ such that
\begin{itemize}
    	\item If $e, e' \in E$ with $e \isucc e'$ 
        then $(\labf(e), \labf(e')) \in D$

	\item If $e, e' \in E$ with $(\labf(e), \labf(e')) \in D$, then $e \leq
        e'$ or $e' \leq e$
\end{itemize}

Let $\traces$ denote the set of all traces over $\dalphabet$.
Henceforth a trace means a trace over $\dalphabet$ unless 
specified otherwise. Let $t = (E, \leq, \labf) \in \traces$.
The elements of $E$ are referred to as \emph{events} in $t$ and for
an event $e$ in $t$, $\loc(e)$ abbreviates $\loc(\labf(e))$. Further, let 
$i \in \pset$. The set of $i$-events in $t$ is $E_i = 
\{e \in E ~|~i \in \loc(e)\}$. This is the set of events in
which process $i$ participates. It is clear that $E_i$ is 
totally ordered by $\leq$.

A subset $c \subseteq E$ is a \emph{configuration} of $t$ if and only if
$\dcset{c} = c$. 
We let $\configs{t}$ denote the set of all configurations of
$t$. Notice that 
$\emptyset$, the empty set, and $E$ are configurations.
More importantly, 
$\dcset{e}$ is a configuration for every $e \in E$. 
There are two 
natural transition relations that one may associate with
the configurations of $t$. The event based transition relation ${\et{t}}
\subseteq \configs{t} \times E \times \configs{t}$ is 
defined by $c \etl{t}{e} c'$ if and only if $e \notin c$ and 
$c \cup \{e\} = c'$. The action based transition relation
${\at{t}} \subseteq \configs{t} \times \alphabet 
\times \configs{t}$ is defined by $c \atl{t}{a} 
c'$ if and only if there exists $e \in E$ such that $\labf(e) = a$ 
and $c \etl{t}{e} c'$.

Now we turn our attention to the important operation of concatenation of traces.
Let $t = (E, \leq, \labf) \in \traces$ and $t' = (E', \leq'
, \labf') \in \traces$. Without loss of generality, we can assume $E$ and $E'$
to be disjoint. We define $tt' \in \traces$ to be 
the trace $(E'', \leq'', \labf'')$ where 
\begin{itemize}
	\item $E'' = E \cup E'$,
	\item $\leq''$ is the transitive closure of 
		${\leq} \cup {\leq'} \cup \{(e,e') \in E \times
			E' ~|~ (\labf(e), \labf'(e')) \in 
		\dep \}$,
	\item $\labf''\colon E'' \to \alphabet$ where 
		$\labf''(e) = \labf(e)$ if $e \in E$;
		otherwise, $\labf''(e) = \labf'(e)$.
\end{itemize}
This operation, henceforth referred to as \emph{trace 
	concatenation}, gives $\traces$ a monoid structure.
Observe that, with $a$ (resp. $b$) denoting the singleton trace with
the only event labelled $a$ (resp. $b$), if $(a,b) \in I$ then
$ab=ba$ in $\traces$.

A basic result in trace theory gives a presentation of the trace monoid
as a quotient of the \emph{free} word monoid $\Sigma^*$. See \cite{traces-book} for more details.
Let ${\sim_I} \subset \Sigma^* 
\times \Sigma^*$ be the congruence generated by $ab \sim_I ba$ 
for $(a,b) \in I$.

\begin{proposition}\label{basicprop} The canonical morphism from
$\Sigma^* \to \traces$, sending a letter $a \in \Sigma$ to the trace $a$, factors through 
the quotient monoid $\Sigma^*/{\sim_I}$  and induces an isomorphism from 
$\Sigma^*/{\sim_I}$ to the trace monoid $\traces$.
\end{proposition}
\subsection{Transformation Monoids and Trace Languages}
A map from a set $X$ to itself is  called a \emph{transformation}
of $X$. Under function composition, the set of all such 
transformations forms a monoid; let us denote this monoid
by $\trans(X)$. The function composition $f_1f_2$ (sometimes also denoted 
$f_1\circ f_2$)
applies from left-to-right, that is, $(f_1f_2)(\cdot) = f_2(f_1(\cdot))$.
% Paul: I think we should avoid \circ since this usually applies from right to 
% left and would be confusing for many readers.

A \emph{transformation monoid} (or simply \emph{tm}) is 
a pair $T=\tmi$ where $M$ is a submonoid of $\trans(X)$. The tm $\tmi$ is called \emph{finite} if $X$ is finite. 

\begin{example}
Consider $X = \{1,2\}$ with the monoid $M = \{\mbox{id}_X, r_1, r_2\}$ 
where $\mbox{id}_X$ is the identity transformation and $r_i$ maps every element in $X$ to element $i$.  
Note that $r_1 r_2 = r_2$ and $r_2 r_1 = r_1$.
Then $(X, M)$ is a tm. We will refer to it as $U_2$.
\end{example}

Let $T=(X,M)$ be a tm. By a morphism $\phi$ from $\traces$ to $T$, we mean
a (monoid) morphism $\phi \colon \traces \rightarrow M$. We abuse the
notation and also write this as $\phi \colon \traces \rightarrow T$.
Observe that, if $(a,b) \in I$, then as $ab=ba$ in $\traces$, $\phi(a)$ and $\phi(b)$
must commute in $M$. In fact, in view of Proposition~\ref{basicprop}, any function
$\phi\colon  \Sigma \rightarrow M$ which has the property that $\phi(a)$ and $\phi(b)$ commute
for every $(a,b) \in I$, can be uniquely
extended to a morphism from $\traces$ to $M$.

Transformation monoids can be naturally used to recognize trace languages. 
Let $L \subseteq \traces$ be a trace language. We say that $L$ is 
\emph{recognized by} a tm $T=(X,M)$ if there exists a morphism $\phi\colon \traces \rightarrow T$,
an \emph{initial} element $x_{\text{in}} \in X$
and a \emph{final} subset $X_{\text{fin}} \subset X$ such that 
$L = \{t \in \traces \mid \phi(t)(x_{\text{in}}) \in X_{\text{fin}}\}$.
A trace language is said to be \emph{recognizable} if it is recognized by a finite tm.

\section{New Algebraic Framework}\label{sec:alg}
\subsection{Asynchronous Transformation Monoids}
Recall that we have a fixed finite set $\pset$ of processes.
If $\pset$ is clear from the context, we use 
the simpler notation $\{X_i\}$ to denote the 
$\pset$-indexed family ${\{X_i\}}_{i \in \pset}$. The elements
of the sets in a $\pset$-indexed family will be typically called \emph{states}.

We begin with some notation involving local and global states.
Suppose that each process $i \in \pset$ is equipped with a finite non-empty 
set of \emph{local $i$-states}, denoted $\ls_i$. We set 
$\ls = \bigcup_{i \in \pset} \ls_i$ and call 
$\ls$ the set of \emph{local states}. We let $P$ range over 
non-empty subsets of $\pset$ and let $i,j$ range over 
$\pset$. A \emph{$P$-state} is a map $s\colon  P \to \ls$ such that 
$s(j) \in \ls_j$ for every $j \in P$. We let 
$\ls_P$ denote the set of all $P$-states. We call 
$\gs$ the set of all \emph{global states}.

If $P' \subset P$ and $s \in \ls_P$ then $s_{P'}$ is
$s$ restricted to $P'$. We use the shorthand ${-\!P}$ to indicate
the complement of $P$ in $\pset$. We 
sometimes split a global state 
$s \in \gs$ as $(s_P, s_{-\!P}) \in \ls_P \times \ls_{-\!P}$.
We let $S_a$ denote the set of all $\loc(a)$-states which we also call $a$-states
for simplicity.
Thus if $\loc(a) \subset P$ and $s$ is a 
$P$-state we shall write $s_a$ to mean $s_{\loc(a)}$. 

Now we are ready to introduce a new class of transformation monoids. 

\begin{definition}
    An \emph{asynchronous transformation monoid} (in short, atm) $T$ (over $\mathcal{P}$)
    is a 
    pair $(\{S_{i}\}, M)$ where
	\begin{itemize}
		\item $S_{i}$ is a finite non-empty set for each process $i \in \pset$.
		\item $M$ is a submonoid of $\trans(\gs)$, the monoid of
            all transformations from $\gs$ to itself. 
	\end{itemize}
\end{definition}

Note that this definition is dependent on $\pset$ and 
an atm $T=(\{S_i\}, M)$ naturally induces the tm $(\gs, M)$.
We abuse the notation and write $T$ also for this tm.

A crucial feature of the definition of an atm is that it makes
a clear distinction between local and global states. Observe that
the underlying transformations operate on global states.
It will be useful to know whether a global transformation is essentially
induced by a particular subset of processes. We develop some notions
to make this precise.

Fix an atm $(\{S_i\}, M)$ and $P \subseteq \pset$.
Let $f\colon  S_{P} \rightarrow S_{P}$ be a map.
We define $g\colon  \gs \rightarrow \gs$ as: for
$s \in \gs$,
\begin{center}
	$g(s) = s'$ iff $f(s_P) = s'_{P}$ and $s_{-\!P} = 
    s'_{-\!P}$
\end{center}
We refer to $g$ as the extension of $f$. More generally, $h\colon  S_{\mathcal{P}}
\rightarrow S_{\mathcal{P}}$ is said to be a $P$-map if it is the extension of some
$f\colon  S_{P} \rightarrow S_{P}$. Note
that, in this case, for all $s=(s_P, s_{-\!P}) \in \gs$,
$h((s_P, s_{-\!P}))= (f(s_P), s_{-\!P})$
and $f$ is uniquely determined by $h$.
It is worth pointing out that a map $h\colon  \gs \rightarrow \gs$
with the property that for every $s \in \gs$ there exists $s'_P \in S_P$
such that $h((s_P, s_{-\!P})) = (s'_P, s_{-\!P})$ is \emph{not} necessarily a $P$-map.
This condition merely says that the $(-P)$-component of a global state
is not updated by $h$. The update of the $P$-component may still depend
on the $(-P)$-component.

The following lemma provides a characterization of $P$-maps. We skip
the easy proof.
\begin{lemma}\label{easylemma} Let $h\colon  \gs \rightarrow \gs$. Then $h$ is a $P$-map
	if and only if for every $s$ in $\gs$, ${[h(s)]}_{-\!P} = s_{-\!P}$ and
	for every $s, s'$ in $\gs$, $s_P=s'_P$ implies that ${[h(s)]}_P = {[h(s')]}_P$.
\end{lemma}

A simple but crucial observation regarding $P$-maps is recorded in the following lemma.
\begin{lemma}\label{simplelemma}%
	Let $f, g \colon  S_{\!\mathcal{P}} \rightarrow S_{\!\mathcal{P}}$ be such that $f$ is
    a $P$-map and $g$ is a $P'$-map. If $P \cap P' = \emptyset$, then $fg=gf$.
\end{lemma}
\begin{proof} Let $f$ (resp. $g$) be the extension of some $f'\colon  S_P \rightarrow S_P$ (resp. $g'\colon S_{P'} \rightarrow S_{P'}$). With $Q=\pset\!-\!(P\cup P')$,
we can split a global state $s \in \gs$ as $s=(s_P,s_{P'},s_Q)$.
In the split notation, we have 
\begin{align*}
fg ~((s_P, s_{P'}, s_Q))  =  g ~((f'(s_P), s_{P'}, s_Q))  =  (f'(s_P), g'(s_{P'}), s_Q) \\
gf ~((s_P, s_{P'}, s_Q))  =  f ~((s_P, g'(s_{P'}), s_Q))  =  (f'(s_P), g'(s_{P'}), s_Q)
\end{align*}
This shows that $f$ and $g$ commute.
\end{proof}

\begin{example}\label{ex:u2l} Fix a process $\ell \in \pset$.
We define the atm $U_2[\ell]= (\{S_i\}, M)$ where, $S_{\ell}=\{1,2\}$ and for each $i \neq \ell$, $S_i$ has exactly one element.
Observe that $S_\pset$ has only two global states which are completely determined by their $\ell$-components.
We will identify a global state with its $\ell$-component.
The monoid $M$ is $\{\mbox{id}_{S_\pset}, r_1, r_2\}$ 
where $\mbox{id}_{S_\pset} $ is the identity transformation and $r_i$ maps every global state to the global state $i$.  
Note that $r_1$ and $r_2$ are $\{\ell\}$-maps.
\end{example}

\subsection{Asynchronous Morphisms} Now we fix a distributed alphabet $\dalphabet = {\{\alphabet_i\}}_{i \in \pset}$ 
over $\pset$ and introduce special
morphisms from the trace monoid $\traces$ to atm's.

\begin{definition}
	Let $T = (\{S_{\! i}\}, M)$ be an atm. An \emph{asynchronous morphism} 
	from
	$TR(\widetilde{\Sigma})$ to $T$ is a (monoid) morphism $\phi \colon
	TR(\widetilde{\Sigma}) \rightarrow M$ such that for $a \in \Sigma$, 
	$\phi(a)$ is an $a$-map. 
\end{definition}
It is important to observe that, contrary to the sequential case, a morphism from $\traces$ to $M$ is not necessarily an asynchronous
morphism from $\traces$ to the atm $T=(\{S_i\},M)$. In a morphism $\psi \colon  \traces \to M$, for $(a,b) \in I$,
$\psi(a)$ and $\psi(b)$ must commute; however $\psi(a)$ (resp. $\psi(b)$) may not 
be an $a$-map (resp.\ $b$-map).

A fundamental result about asynchronous morphisms is stated in the following lemma. 
\begin{lemma}\label{lem:extends-asyn-morph}%
	Let $T = ( \{S_{\! i}\}, M)$ be an atm. Further, let $\phi \colon \Sigma \to M$ 
	be such that, for $a \in \Sigma$, $\phi(a)$ is an $a$-map. Then $\phi$ can 
	be uniquely extended to an asynchronous morphism from
	$\traces$ to $T$.
\end{lemma}
	\begin{proof} 
As the word monoid $\Sigma^*$ is the \emph{free} monoid generated by $\Sigma$, the map 
$\phi$ uniquely extends to a morphism from $\Sigma^*$ to $M$. 
Recall that, by Proposition~\ref{basicprop}, $\traces$ is the 
	quotient of $\Sigma^*$ by the relations
	of the form $ab=ba$ where $(a,b) \in I$. Therefore, in order to complete
	the proof, we simply need to show that $\phi(a)$ and $\phi(b)$ commute.
	If $(a,b) \in I$, then $\loc(a) \cap \loc(b) = \emptyset$.
	As $\phi(a)$ is an $a$-map and $\phi(b)$ is a $b$-map,
	by Lemma~\ref{simplelemma}, $\phi(a)$ and $\phi(b)$ commute.
\end{proof}

\begin{example}\label{ex:amorph}
Consider $\dalphabet = \{ \s_{p_1} = \{a,b\}, \s_{p_2} = \{b,c\}, \s_{p_3} = \{c\} \}$. A function 
$\varphi(a) = r_1$, $\varphi(b) = r_2$ and $\varphi(c) = \mathrm{id}$ extends to an asynchronous morphism 
from $\traces$ to $U_2[p_1]$.
\end{example}

Now we extend the notion of trace-language recognition from tm's to atm's via asynchronous
morphisms.
Let $L \subseteq \traces$ be a trace language. We say that $L$ is recognized by an 
atm $T = (\{S_{i}\}, M)$ if there exists an \emph{asynchronous} morphism $\phi\colon 
\traces \rightarrow T$, an \emph{initial} element $s_{\text{in}} \in \gs$ and a 
\emph{final} subset $\gf
\subseteq \gs$ such that 
\begin{center}
$L = \{t \in \traces \mid \phi(t)(s_{\text{in}}) \in \gf  \}$
\end{center}
In the rest of this subsection, we bring out the intimate relationship between
asynchronous morphisms and asynchronous automata.
We begin with the description of an asynchronous automaton -- a
model introduced by Zielonka for concurrent computation on
traces. 

An \emph{asynchronous automaton} $\aaa$ over $\dalphabet$ is a 
structure $ ( {\{\ls_i\}}_{i \in \pset}, \ab  {\{\lt_a\}}_{a \in \alphabet}, s_{\text{in}} ) $ where
\begin{itemize}
\item $\ls_i$ is a finite non-empty set of local $i$-states for each process $i$
\item For $a \in \alphabet$, $\lt_a \colon  \ls_a \to \ls_a$ is a transition 
function on $a$-states
\item $s_{\text{in}} \in \gs$ is an initial global state
\end{itemize}
Observe that an $a$-transition of $\aaa$ reads and updates only the local states of the agents which participate in $a$.
As a result, actions which involve disjoint sets of agents are processed concurrently by $\aaa$.
For $a \in \alphabet$, let $\gt_a \colon  \gs \to \gs$ be the extension of $\lt_a \colon  \ls_a \to \ls_a$. Clearly,
if $(a,b) \in I$ then $\Delta_a$ and $\Delta_b$ commute.
Similar to $\pset$-indexed families, we will follow the convention of writing $\{Y_a\}$ to denote the
$\alphabet$-indexed family ${\{Y_a\}}_{a \in \alphabet}$. 
 
Now we describe the notion of a run of $\aaa$ on an input trace. A trace run is easiest to define using
configurations. Towards this, fix a trace $t = (E, \leq, \labf) \in \traces$.
Recall that (Section~\ref{section:traces})  $\configs{t}$  is the set of all configurations of $t$ and
${\at{t}} \subset \configs{t} \times \alphabet \times \configs{t}$ is the natural
action based transition relation on configurations.
 A \emph{trace run} of $\aaa$ over $t \in \traces$ is a map 
 $\run \colon \configs{t} \to \gs$ such that $\run(\emptyset) = s_{\text{in}}$,
 and for every $(c, a, c')$ in $\at{t}$, we have $\gt_a(\run(c)) = \run(c')$.
 As $\aaa$ is deterministic, 
$t$ admits a unique trace run; it will be denoted by $\run_t$.
 
Let $L \subset \traces$ be a trace language. We say that $L$ is \emph{accepted} by $\aaa$
 if there exists a subset $\gf \subset \gs$ of final
 global states such that 
 $L = \{t = (E, \leq, \labf) \in \traces \mid \run_t(E) 
 \in \gf \}$.

Our aim is to associate with $\aaa$, a natural atm $T_\aaa$ and an asynchronous morphism $\phi_\aaa$
such that languages accepted by $\aaa$ are precisely the languages recognized via $\phi_\aaa$.

We first describe the \emph{transition} monoid $M_\aaa$ associated to $\aaa$. It is possible to
 extend the global transition functions $\{ \gt_a \}$ to arbitrary traces using
Proposition~\ref{basicprop}. For 
 $t \in \traces$, we denote this extended global transition function by $\gt_t \colon \gs\to \gs$.
 It is easy to check that, for $t = 
 (E, \leq, \labf) \in \traces$, $\gt_t(s_{\text{in}}) = \run_t(E)$.
Further, as expected, for $t,t' \in \traces$, the function composition $\gt_t 
 \gt_{t'}$ is identical to $\gt_{tt'}$. We let 
$M_\aaa$ be the finite set of functions $\{\gt_t \mid t \in \traces\}$.
Clearly, it is a monoid under the usual composition of functions.

Next, we define the \emph{transition} atm of $\aaa$ to be 
$T_\aaa = (\{\ls_i\},\ab M_\aaa)$ and the natural map $\phi_\aaa \colon \traces \to M_\aaa$ sending $t$ to $\gt_t$. 
It is clear that $\phi_\aaa$ is a morphism of monoids. Furthermore, it is an asynchronous morphism
from $\traces$ to $T_\aaa$; this is because, for $a \in \Sigma$, $\phi_\aaa(a) = \gt_a$ is in fact an $a$-map
of the atm $T_\aaa$. 
The map $\phi_\aaa$ is called the \emph{transition} (asynchronous) morphism of $\asa$.
Note that, for $t = (E, \leq, \labf) \in \traces$, 
\begin{center}
$\phi_\aaa(t)(s_{\text{in}}) = \gt_t(s_{\text{in}}) = \run_t(E)$
\end{center}
We refer to the above statement as the \emph{duality} between a run of $\aaa$ and an evaluation of $\phi_\aaa$.

The following lemma summarizes the above discussion for later reference and its proof is immediate.
  \begin{lemma}\label{lem:atmza}%
	  Given an asynchronous automaton $\aaa = (\{\ls_i \},\ab \{\lt_a\}, s_{\text{in}})$ over $\dalphabet$,
  the transition atm $T_\aaa = (\{\ls_i\}, M_\aaa)$ and the transition asynchronous morphism
  $\phi_\aaa \colon \traces \to T_\aaa$ are effectively constructible. Moreover, if $L$ 
  is a trace language, then $L$ is accepted by $\aaa$ 
  if and only if it is recognized 
  by $T_\aaa$ via $\phi_\aaa$ with $s_{\text{in}}$ as the initial state.
  \end{lemma}

We now provide a form of converse to Lemma~\ref{lem:atmza}. 
Towards this, we fix an atm $T=\atma$, a state $s_{\text{in}} \in \gs$ and an 
	  asynchronous morphism $\phi \colon \traces \to T$.
Since $\phi$ is an asynchronous 
 morphism, $\phi(a)$ is an $a$-map, and is an extension of 
 some $\lt_a\colon  \ls_a \to \ls_a$ over $a$-states. 
 We set $\aaa_\phi = (\{\ls_i\}, \{\lt_a\}, s_{\text{in}})$ over $\dalphabet$.
 It turns out that the transition monoid of $\aaa_\phi$
is the image of $\phi$, a submonoid of $M$ and the transition morphism
of $\aaa_\phi$ is the appropriate restriction of $\phi$ to this submonoid.
The next lemma is 
easy to prove and we skip its proof.

\begin{lemma}\label{lem:zaatm}
Given $T=\atma$, $\phi \colon \traces \to T$ and $s_{\text{in}} \in \gs$,
the asynchronous automaton $\aaa_\phi$ 
over $\dalphabet$ is effectively constructible. Moreover,
	  a trace language $L \subset \traces$
	  is recognized by $T$ via $\phi$ (with initial state $s_{\text{in}}$) if and
	  only if it is 
	  accepted by $\aaa_\phi$.
  \end{lemma} 
\subsection{Asynchronous Wreath Product}
We begin with the crucial definition of a wreath product of transformation monoids. 
For sets $U$ and $V$, we denote the set of all functions from $U$ to $V$ by 
${\mathcal F}(U, V)$.

\begin{definition} [Wreath Product]\label{def:wp}
Let $T_1 = (X, M)$ and $T_2 = (Y, N)$ be two tm's.
We define $T = T_1 \wr T_2$ to be the tm $(X \times Y, M \times {\mathcal F}(X, N))$
where, for $m \in M$ and $f \in {\mathcal F}(X, N)$, $(m,f)$ 
represents the following transformation on $X \times Y$:
\begin{center}
$\text{~for~} (x,y) \in \wpis, \;\;\;(m,f)((x,y)) = (m(x), f(x)(y))$
\end{center}
The tm $T$ is called the \emph{wreath product} of $T_1$ and $T_2$.
It turns out that, for $(m_1, f_1), (m_2, f_2)$ in 
$M \times {\mathcal F}(X, N)$, the composition law $(m_1, f_1)(m_2,f_2) = (m,f)$ is 
such that $m = m_1m_2$ and for $x \in X, f(x) = f_1(x) + 
f_2(m_1(x))$. Here $+$ denotes the composition operation of $N$.
\end{definition}
It is a standard fact that the wreath product operation is associative~\cite{Eilenberg1976AutomataLA}.
We now adapt this operation to asynchronous transformation monoids.
\begin{definition}
Let $T_1 = (\{S_i\},M)$ and $T_2 = (\{Q_i\},N)$ be two atm's. 
We define their asynchronous wreath product, also denoted by $T_1 \wr T_2$, 
to be the atm $(\{S_i \times Q_i\}, M \times {\mathcal F}(\gs, N))$. 
An element $(m,f) \in  M \times {\mathcal F}(\gs, N)$
represents the following global\footnote{a global state (resp.\ $P$-state) of $T_1\wr T_2$ is canonically 
	identified with an element of 
$S_\pset \times Q_\pset$ (resp.\ $S_P \times Q_P$)} transformation on $S_\pset \times Q_\pset$:
\begin{center}
$\text{~for~} (s,q) \in S_\pset \times Q_\pset , \;\;\;(m,f)((s,q)) = (m(s), f(s)(q))$
\end{center}
The composition law on $M \times {\mathcal F}(\gs, N)$ is the same as in Definition~\ref{def:wp}.
\end{definition}
An important observation is that the tm associated with $T_1 \wr T_2$ is the
wreath product of the tms $(\gs,M)$ and $(Q_\pset,N)$ associated
with $T_1$ and $T_2$ respectively. Sometimes, we will refer to the asynchronous wreath product simply as wreath product. 
The associativity of the asynchronous wreath product operation follows immediately.

We now present an important combinatorial lemma regarding 
the `support' of a global transformation in the wreath product. 
It plays a crucial role later.

\begin{restatable}{lemma}{combinatorial}\label{lem:combinatorial}%
Fix atms $T_1 = (\{S_i\},M)$ and $T_2 = (\{Q_i\},N)$. 
Let $(m,f) \in M \times \trans(\gs, N)$ represent a $P$-map in 
	$T_1 \wr T_2$ for some subset $P \subset \pset$. Then
	\begin{itemize}
		\item $m$ is a $P$-map in $T_1$.
		\item For every $s \in \gs, ~f(s)$ is a $P$-map in $T_2$. Further,
		if $s, s' \in \gs$ are such that $s_P = s'_P$, then $f(s) = f(s')$.
	\end{itemize}
	\end{restatable}
\begin{proof}
	Fix $x_0 \in S_{\!-\!P}$ and $y_0 \in Q_{\!-\!P}$. We
	define
	$g_1\colon  S_P \to S_P$ and $g_2\colon  S_P \times Q_P \to Q_P$ by 
    $g_1(x) = {[m((x, x_0))]}_P$ and $g_2(x,y) = {[f((x,x_0))(y, y_0)]}_P$.
	We first show that 
	for all $s \in \gs, q \in Q_{\pset}$, 	
	$(m,f)((s,q)) =( (g_1(s_P), s_{-\!P}), (g_2(s_P, q_P), q_{-\!P}) )$.
    Take an arbitrary $(s,q) \in \gs \times Q_{\pset}$. Then consider
     the global state $((s_P,{x_0}), (q_P, {y_0}))$ 
     sharing the same $P$-component as $(s,q)$ and the fixed $-\!P$-component $(x_0,y_0)$. By the wreath product action
(see Definition~\ref{def:wp}),
$(m,f) \left(((s_P,{x_0}), (q_P, {y_0}))\right) = (m((s_P,{x_0})), f((s_P,{x_0}))((q_P, {y_0})))$.
Being 
    a $P$-map, $(m,f)$ does not change the $-\!P$-component of any global state. So we have $m((s_P,{x_0})) = 
    ({[m((s_P,{x_0}))]}_P, {x_0} )$ and 
	$f((s_P,{x_0}))((q_P, {y_0})) = 
	({[f((s_P,{x_0}))((q_P, {y_0}))]}_P, {y_0})$.
	
	Let $(m,f)((s,q)) = (s',q')$. Since $(m,f)$ is a $P$-map and 
	the two global states $(s,q)$ and $((s_P,{x_0}), (q_P, {y_0}))$
share the same
	 $P$-component, by Lemma~\ref{easylemma},
	 $s'_P = {[m((s_P,{x_0}))]}_P$ and 
	 $q'_P = {[f((s_P,{x_0}))((q_P, {y_0}))]}_P$. Further, 
$s'_{-\!P} = s_{-\!P}$ and $q'_{-\!P} = q_{-\!P}$.
Using the definitions of $g_1$ and $g_2$, we immediately
see that 
$(m,f)((s,q)) =\left( (g_1(s_P), s_{-\!P}), (g_2(s_P, q_P), q_{-\!P}) \right)$.
However, by the wreath product action,
$(m,f)((s,q))= (m(s), f(s)(q))$. Comparing this with the previous expression,
we have $m(s) = (g_1(s_P), s_{-\!P})$ and
$f(s)(q) = (g_2(s_P, q_P),\ab q_{-\!P})$.
The result now follows from Lemma~\ref{easylemma}.
\end{proof}

\subsection{Local Wreath Product Principle}\label{sec:wpp}
Let $\aaa = (\{\ls_i\}, \{\lt_a\}, s_{in})$ be an asynchronous 
automaton over $\dalphabet$. Based on $\aaa$ and 
$\dalphabet$, we define the 
alphabet
$\alphabet^{\|_S} =  \{ (a, s_a) \mid a \in \alphabet, s \in \gs \}$
where a letter $a$ in $\alphabet$ is extended with local $a$-state
 information of $\aaa$. This can naturally be viewed as a distributed
  alphabet $\widetilde{\tralphabet}$ where $\forall a \in \alphabet, \forall s \in \gs$,
  $(a, s_a) \in \tralphabet_{i}$ if and only if $a \in \alphabet_i$.
  Then $\aaa$ induces the following transducer over traces.

\begin{definition}[Local Asynchronous Transducer]
	Let $\chi_{\aaa} \colon \ab  \traces \to \tracesdtr$ be defined as follows. 
	If $t = (E,\leq, \labf) \in \traces$, then $\chi_{\aaa}(t) = t'$
	where $t' = (E, \leq, \mu) \in \tracesdtr$ with the labelling 
	$\mu\colon  E \to \tralphabet$ defined by:
	\begin{center}
	$\forall e \in E, \mu(e) = (a,s_a) \text{ where } a = \labf(e) 
	\text{ and } s = \run_t(\dcset{e}\setminus \{e\})$
	\end{center}
	(recall that $\run_t$ is the unique trace run of $\aaa$ over $t$).
	  We call $\chi_{\aaa}$ the local asynchronous 
	 transducer of $\aaa$.
\end{definition}
\begin{example}\label{ex:lat}
	Let $\chi$ be the local asynchronous transducer associated to $A_\varphi$ where $\varphi$ is as in Example~\ref{ex:amorph}. Figure~\ref{fig:lat-example} shows
	the run of $A_\varphi$ on a trace $t \in \traces$ (by showing local process states before and after each event),
	and the resulting trace $\chi(t) \in
	\tracesdtr$.
	\begin{figure}[ht]
		\centering
		\begin{tikzpicture}
\draw (0,0) -- (0, 5*0.2 + 3*0.5);
%\node at (-0.2, 0.2 + 0.5/2) {$2$};
%\node at (-0.2, 2*0.2 + 0.5 + 0.5/2) {$1$};

\draw (0.3,1.6) rectangle (0.8,2.1);
\draw (1.3,0.9) rectangle (1.8,2.1);
\draw (2.3,0.2) rectangle (2.8,1.4);

\draw (0,0.2 + 0.5/2) -- (2.3, 0.2 +0.5/2); 
\draw (2.8,0.2 + 0.5/2) -- (3.0,0.2 + 0.5/2); 

\draw (0,2*0.2+0.5+0.5/2) -- (1.3,2*0.2+0.5+0.5/2); 
\draw (1.8,2*0.2+0.5+0.5/2) -- (2.3,2*0.2+0.5+0.5/2); 
\draw (2.8,2*0.2+0.5+0.5/2) -- (3.0,2*0.2+0.5+0.5/2); 

\draw (0, 3*0.2 + 2*0.5 + 0.5/2) -- (0.3, 3*0.2 + 2*0.5 + 0.5/2);
\draw (0.8, 3*0.2 + 2*0.5 + 0.5/2) -- (1.3, 3*0.2 + 2*0.5 + 0.5/2);
\draw (1.8, 3*0.2 + 2*0.5 + 0.5/2) -- (3.0, 3*0.2 + 2*0.5 + 0.5/2);

\node at (0.3 + 0.5/2, 1.6 + 0.5/2) {$a$};
\node at (1.3 + 0.5/2, 0.9 + 0.5 + 0.2/2) {$b$};
\node at (2.3 +  0.5/2, 0.2 + 0.5 + 0.2/2) {$c$};

\node at (-0.2, 0.2 + 0.5/2) {$p_3$};
\node at (-0.2, 2*0.2 + 0.5 + 0.5/2) {$p_2$};
\node at (-0.2, 3*0.2 + 2*0.5 + 0.5/2) {$p_1$};

\node at (0.15, 3*0.2 + 2*0.5 +0.5/2 + 0.15) {\tiny $1$};
\node at (1.15, 3*0.2 + 2*0.5 +0.5/2 + 0.15) {\tiny $1$};
\node at (1.95, 3*0.2 + 2*0.5 +0.5/2 + 0.15) {\tiny $2$};

\node at (1.1, 2*0.2 + 0.5 +0.5/2 + 0.15) {\tiny $\bot_2$};
\node at (2.1, 2*0.2 + 0.5 +0.5/2 + 0.15) {\tiny $\bot_2$};
\node at (2.99, 2*0.2 + 0.5 +0.5/2 + 0.15) {\tiny $\bot_2$};

\node at (2.1, 0.2 + 0.5/2 + 0.15) {\tiny $\bot_3$};
\node at (2.99, 0.2 + 0.5/2 + 0.15) {\tiny $\bot_3$};

\node at (1.5, -0.4) {\small Run of trace $t$ in $A_\varphi$};

\begin{scope}[shift={(5,0)}]
\draw (0,0) -- (0, 5*0.2 + 3*0.5);
%\node at (-0.2, 0.2 + 0.5/2) {$2$};
%\node at (-0.2, 2*0.2 + 0.5 + 0.5/2) {$1$};

\draw (0.3,1.6) rectangle (0.8,2.1);
\draw (1.3,0.9) rectangle (1.8,2.1);
\draw (2.3,0.2) rectangle (2.8,1.4);

\draw (0,0.2 + 0.5/2) -- (2.3, 0.2 +0.5/2); 
\draw (2.8,0.2 + 0.5/2) -- (3.0,0.2 + 0.5/2); 

\draw (0,2*0.2+0.5+0.5/2) -- (1.3,2*0.2+0.5+0.5/2); 
\draw (1.8,2*0.2+0.5+0.5/2) -- (2.3,2*0.2+0.5+0.5/2); 
\draw (2.8,2*0.2+0.5+0.5/2) -- (3.0,2*0.2+0.5+0.5/2); 

\draw (0, 3*0.2 + 2*0.5 + 0.5/2) -- (0.3, 3*0.2 + 2*0.5 + 0.5/2);
\draw (0.8, 3*0.2 + 2*0.5 + 0.5/2) -- (1.3, 3*0.2 + 2*0.5 + 0.5/2);
\draw (1.8, 3*0.2 + 2*0.5 + 0.5/2) -- (3.0, 3*0.2 + 2*0.5 + 0.5/2);

\node at (0.3 + 0.5/2, 1.6 + 0.5/2 + 0.1) {\small $a$}; % \vspace{1mm} $1$};
\node at (0.3 + 0.5/2, 1.6 + 0.5/2 - 0.1) {\tiny $1$}; % \vspace{1mm} $1$};
\node at (1.3 + 0.5/2, 1.6 + 0.5/2) {\small $b$};
\node at (1.3 + 0.5/2, 0.9 + 0.5 + 0.2/2) {\tiny $1$};
\node at (1.3 + 0.5/2, 0.9 + 0.5/2) {\tiny $\bot_2$};
\node at (2.3 +  0.5/2, 0.2 + 0.5 + 0.2 +0.5/2) {\small $c$};
\node at (2.3 +  0.5/2, 0.2 + 0.5 + 0.2/2) {\tiny $\bot_2$};
\node at (2.3 +  0.5/2, 0.2 + 0.5/2) {\tiny $\bot_3$};

\node at (-0.2, 0.2 + 0.5/2) {$p_3$};
\node at (-0.2, 2*0.2 + 0.5 + 0.5/2) {$p_2$};
\node at (-0.2, 3*0.2 + 2*0.5 + 0.5/2) {$p_1$};

\node at (1.5, -0.4) {\small Trace $\chi(t)$};
\end{scope}
\end{tikzpicture}
		\caption{Local asynchronous transducer output on a trace; $S_{p_2} = \{\bot_2\}, S_{p_3}= \{\bot_3\}$.}%
		\label{fig:lat-example}
	\end{figure}
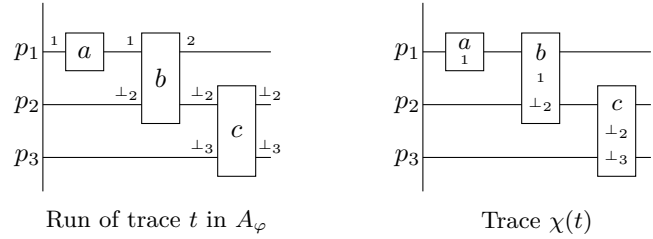
	\end{example}

Note that, in general, $\chi_\aaa$ is not a morphism of
monoids.
The following lemma is a straightforward consequence of the definition of
$\chi_A$ and the duality between trace runs of $\aaa$ and
evaluations of the asynchronous asynchronous morphism $\phi_A$.
 
 \begin{lemma}\label{lem:trfactor}
	 Let $t \in \traces$ with factorization $t = t'a$ (where $a \in \dalphabet$), and 
	$s = \phi_A(t')(s_{\text{in}})$. Then the trace $\chi_A(t) \in \tracesdtr$
	factors as $\chi_A(t) = 
 	\chi_A(t')(a,s_a)$.
 \end{lemma}

 \begin{restatable}{theorem}{wppfirst}\label{wpp1}%
Let $A$ be an asynchronous automaton over $\dalphabet$ and $\chi_A$ be the corresponding local asynchronous transducer.
If $L \subseteq \tracesdtr$ is recognized by an atm $T$, then $\chi_A^{-1}(L)$ is recognized by the atm $T_A \wr T$.
\end{restatable}
\begin{proof}
	Let $\psi \colon  \tracesdtr \to T = (\{Q_i\}, N)$ be an asynchronous morphism,
	which recognizes $L$ with $q_{\text{in}} \in Q_\pset$ as the initial
	global state, and $Q_{\text{fin}} \subseteq Q_\pset$ as the set of
	final global states. Then $L = \{ t \in \tracesdtr \mid 
	\psi(t)(q_{\text{in}}) \in Q_{\text{fin}} \}$. Note that, for $(a,s_a) \in 
    \tralphabet$, $\psi((a, s_a))$ is an $a$-map (that is, an extension of a map from $Q_a$ to 
	$Q_a$; recall that $\loc((a,s_a)) = \loc(a)$).
	
    For $a \in \alphabet$, we set
	$\eta(a) = (\phi_A(a), f_a)$ where $f_a\colon \gs \to N$ is defined by 
	$f_a(s) = \psi((a,s_a))$. It is easy to check that $\eta(a)$ is
	an $a$-map (that is, an extension of a map from $\ls_a \times Q_a$ to $\ls_a \times Q_a$).
    By Lemma~\ref{lem:extends-asyn-morph}, this uniquely defines an asynchronous morphism
    $\eta \colon \traces \to T_A \wr T$.
	
	Let $t = (E, \leq, \labf) \in \traces$. We write $\eta(t) =
	(\pi_1(t),\pi_2(t) )$. It follows from the definition of 
	wreath product that $\pi_1(t) = \phi_A(t)$. Now we claim 
	that $\pi_2(t)(s_{\text{in}}) = \psi(\chi_A(t))$. We prove this by
	induction on the cardinality of $E$. Suppose $t = t'a$. Then $\eta(t) = \eta(t')\eta(a)$.
	As a result, we have  $(\pi_1(t), \pi_2(t))  = (\pi_1(t'), \pi_2(t'))(\pi_1(a), \pi_2(a))$.
	Therefore, for $s \in \gs$, $\pi_2(t)(s) = \pi_2(t')(s) + 
	\pi_2(a)(\pi_1(t')(s))$. In particular, it holds with $s = s_{\text{in}}$. Recall that $\pi_1(t')
	= \phi_A(t')$. Also, by induction, $\pi_2(t')(s_{\text{in}}) = \psi(\chi_A(t'))$.
	Hence, with $s = \phi_A(t')(s_{\text{in}})$,
	\begin{align*}
	\pi_2(t)(s_{\text{in}}) &= \pi_2(t')(s_{\text{in}}) + \pi_2(a)(\pi_1(t')(s_{\text{in}})) \\
	&= \psi(\chi_A(t')) + \pi_2(a)(s) \\
	&= \psi(\chi_A(t')) + \psi((a,s_a)) \\
	&= \psi(\chi_A(t')(a,s_a)) \\
	&= \psi(\chi_A(t))
	\end{align*}
	The last equality follows from Lemma~\ref{lem:trfactor}. So, 
	$t \in \chi_A^{-1}(L)$ if and only if $\chi_A(t) \in L$ if and only if $\psi(\chi_A(t))(q_{\text{in}}) 
	\in Q_{\text{fin}}$ if and only if $\pi_2(t)(s_{\text{in}})(q_{\text{in}}) \in Q_{\text{fin}}$ if and only if 
	$\eta(t)(s_{\text{in}}, q_{\text{in}}) \in \gs \times Q_{\text{fin}}$. This shows that
	$\eta$ recognizes $\chi_A^{-1}(L)$ with $(s_{\text{in}}, q_{\text{in}}) \in 
	\gs \times Q_\pset$ as the initial global state, and $\gs \times 
	Q_{\text{fin}} \subseteq \gs \times Q_\pset$ as the set of final global
	states.
\end{proof}
Now we focus our attention on what is usually termed as the wreath product principle.
 
\begin{restatable}{theorem}{wppsecond}\label{thm:wpp2}%
 Let $T_1$ and $T_2$ be atms and let $L \subseteq \traces$ be a trace language recognized by an asynchronous morphism
 $\eta\colon \traces \to T_1 \wr T_2$, with initial global state $(s_{\text{in}}, q_{\text{in}})$. For each $a\in \alphabet$,
 let $\eta(a) = (m_a,f_a)$. Then $\phi\colon \traces \to T_1$, defined by $\phi(a) = m_a$, is an asynchronous morphism. 
 Finally, let $A = A_\phi$ be the asynchronous automaton associated to $\phi$ and $s_\text{in}$, and let $\chi_A$
 be the corresponding local asynchronous transducer. Then $L$ is a finite union of languages of the form $U \cap \chi_A^{-1}(V)$, where 
 $U \subseteq TR(\tilde\Sigma)$ is recognized by $T_1$, and $V \subseteq \tracesdtr$ is recognized by $T_2$.
 \end{restatable}
\begin{proof}
	We write $T_1 = 
	(\{ \ls_i \},M)$ and $T_2 = ( \{Q_i\},N )$. 
	Consider $a \in \Sigma$ and the $a$-map $\eta(a) = (m_a, f_a) \in M \times \trans(\gs, N)$.
	This means that $\eta(a)$ is an extension of a map from $S_a \times Q_a$ to $S_a \times Q_a$.
	By Lemma~\ref{lem:combinatorial}, $m_a \in M$ is an $a$-map (of $T_1$) and $f_a\colon  S_\pset \to N$ is
	such that, for $s \in S_\pset$, $f_a(s) \in N$ is an $a$-map (of $T_2$) and it depends only on $s_a$.
	In particular, $f_a\colon  S_\pset \to N$ may be viewed as $f_a\colon  S_a \to N$. Below we will use
	$f_a$ in this sense.
	
	Now we define an asynchronous morphism 
	$\psi\colon  \tracesdtr \to T_2$ as follows:
	$\psi((a,s_a)) = f_a(s_a)$.
	Note that, by Lemma~\ref{lem:extends-asyn-morph}, $\psi$ is indeed an asynchronous morphism as $f_a(s_a)$ is an $a$-map. Further,
	as $m_a$ is an $a$-map, $\phi\colon \traces \to T_1$, defined by $\phi(a) = m_a$, also extends to an asynchronous morphism.
	
	Our aim is to express $L$ in terms of languages recognized by $T_1$ and $T_2$.
	It suffices to show the result when $L$ is recognized with a
	single final global state, say 
	$(s_{\text{fin}}, q_{\text{fin}})$. Then $L = \{t \in \traces \mid 
	\eta(t)((s_{\text{in}}, q_{\text{in}})) = (s_{\text{fin}}, q_{\text{fin}}) \}$.
	
	For $t \in \traces$, we write $\eta(t) = (\pi_1(t), \pi_2(t))$. 
	It follows from the definition of $\phi$ that $\phi(t)=\pi_1(t)$.
	Hence, we can 
	alternatively write $L$ as \[L = \{t \in \traces \mid 
	\phi(t)(s_{\text{in}}) = s_{\text{fin}} \text{ and } \pi_2(t)(s_{\text{in}})(q_{\text{in}}) = q_{\text{fin}}
	\}\]
	Let $U = \{t \in \traces \mid \phi(t)(s_{\text{in}}) = s_{\text{fin}}\}$. Then, with
	$W = \{t \in \traces \mid \pi_2(t)(s_{\text{in}})(q_{\text{in}}) = q_{\text{fin}}\}$, 
	$L = U \cap W$.
	By using essentially the same ideas as in the proof of 
	Theorem~\ref{wpp1}, we can show that $\pi_2(t)(s_{\text{in}}) = 
	\psi(\chi_A(t))$. Therefore,
	$W = \{t \in \traces \mid \psi(\chi_A(t))(q_{\text{in}}) = q_{\text{fin}} \}$.
	
	It follows that, with $V = \{ t' \in \tracesdtr \mid \psi(t')(q_{\text{in}}) =
	q_{\text{fin}} \}$, $W = \chi_A^{-1}(V)$. Clearly, $U$ is recognized by 
	the atm $T_1$ (via $\phi$), $V$ is recognized by the atm $T_2$
	(via $\psi$) and $L = U \cap \chi_A^{-1}(V)$. This completes the proof.
\end{proof}

\section{Towards a Decomposition Result}\label{sec:dec}
In this section, we use the algebraic framework developed so far 
to propose an analogue of the fundamental Krohn-Rhodes decomposition theorem 
over traces. We first recall the Krohn-Rhodes theorem in the purely algebraic setting
of transformation monoids. We briefly explain how it is used to
analyze/decompose morphisms from the free monoid and point out some difficulties
that arise when we consider morphisms from the trace monoid.

Let $M$ and $N$ be monoids. We say that $M$ divides $N$ 
(in notation, $M \divides N$) 
if $M$ is a homomorphic image of some submonoid of $N$.
This notion can be extended to transformation monoids.
Let $\tmi$ and $\tmii$ be two tm's. We say that $\tmi$
\emph{divides} $\tmii$, denoted $\tmi \divides \tmii$, if
 there exists a pair of mappings
$(f, \phi)$ where $f \colon Y \rightarrow X$ is a surjective 
function and $\phi\colon  N' \rightarrow M$ is a surjective morphism 
from a submonoid $N'$ of $N$, such that
 $\phi(n)(f(y)) = f(n(y))$ for all $n \in N'$ and all $y \in Y$. 

Recall that $U_2= (\{1,2\}, \{\mbox{id}, r_1, r_2\})$ denotes the \emph{reset} transformation monoid
on two elements. Along with it, the following class of transformation monoids plays
an important role in the Krohn Rhodes theorem.

\begin{example}\label{ex:g}
Let $G$ be a group. Then $(G,G)$ is a tm where the monoid element
$g$ \emph{represents} the transformation $m_g\colon G \rightarrow G$ of the set 
$G$, which is
the right multiplication by $g$. In other words, for $h \in G, m_g(h) = hg$.

\end{example}

We are now in a position to state the Krohn-Rhodes theorem \cite{KR}.
See \cite{str_cirBook} for a classical proof of the theorem, and \cite{DBLP:journals/fuin/DiekertKS12}
for a modern proof.

\begin{theorem}[Krohn-Rhodes Theorem]
Every finite transformation monoid $T=(X,M)$ divides a wreath
product of the form $ T_1 \wr \ldots \wr T_n$
where each factor $T_i$ is either $U_2$ or
is of the form $(G,G)$ for some non-trivial simple group
$G$ dividing $M$.
\end{theorem}

Henceforth, we will be dealing with only finite tms and sometimes
we will omit the qualifier `finite'. 
Now we turn our attention to the use of this decomposition theorem
for analysing word languages recognized by morphisms from the free monoid.

\begin{definition}\label{wordsimulation}
Let $\phi\colon  \Sigma^* \to T=(X,M)$ be a morphism.
Further, let $\psi\colon  \Sigma^* \to T'=(Y,N)$ be another morphism.
We say that $\psi$ \emph{simulates} $\phi$ if there exists
a surjective function $f\colon  Y \to X$ such that, for all $a \in \Sigma$ and all $y \in Y$,
$f(\psi(a)(y)) = \phi(a)(f(y))$.
\end{definition}
\begin{figure}[ht]
		\centering
		\begin{tikzpicture}
			\node  at (0,0) (sw) {$X$};
			\node  at (1.4,0) (se) {$X$};
			\node  at (0,1.4) (nw) {$Y$};
			\node  at (1.4,1.4) (ne) {$Y$};
			\draw[->] (sw.east) to node[below] {\small $\varphi(a)$} (se.west); 
			\draw[->] (nw.east) to node[above] {\small $\psi(a)$}(ne.west); 
			\draw[->] (nw.south) to node[left] {\small $f$} (sw.north); 
			\draw[->] (ne.south) to node[right] {\small $f$} (se.north); 
		\end{tikzpicture}
        \caption{Visual illustration of condition $f(\psi(a)(y)) = \phi(a)(f(y))$ 
        in Definition~\ref{wordsimulation}}
		\label{fig:commute-word-simulation}
	\end{figure}
Observe that if $\psi$ simulates $\phi$ then a language recognized by $\phi$
is also recognized by $\psi$.

\begin{proposition}\label{wordkr} 
	Let $\phi\colon  \Sigma^* \to T=(X,M)$ be a morphism.
	Then there exists a morphism $\psi\colon  \Sigma^* \to T'$ 
        which simulates $\phi$ such that 
the tm $T'$ is of 
the form $T_1 \wr \ldots \wr T_n$
where each factor $T_i$ is either $U_2$ or
$(G,G)$ for some non-trivial simple group
$G$ dividing $M$.
\end{proposition}

\begin{proof}
Given $T$, we get $T'=T_1 \wr \ldots \wr T_n =(Y,N)$ by the Krohn-Rhodes theorem. 
Since $T \divides T'$, there exists a pair of mappings
$(f, \theta)$ where $f \colon Y \rightarrow X$ is a surjective 
function and $\theta\colon  N' \rightarrow M$ is a surjective morphism 
from a submonoid $N'$ of $N$, such that
$\theta(n)(f(y)) = f(n(y))$ for all $n \in N'$ and all $y \in Y$. 
Construct $\psi\colon  \Sigma \to N$ by mapping 
$\psi(a)$, for each $a$ in $\alphabet$, to an arbitrary element in $\theta^{-1}(\phi(a))$. 
Thanks to the fact that $\Sigma^*$ is a free monoid, $\psi$ uniquely extends to
a morphism $\psi\colon  \Sigma^* \to T'$. It is easily checked that
$\psi$ simulates $\phi$.
\end{proof}
Combined with the wreath product principle, the above proposition provides a 
powerful inductive tool to prove many non-trivial results in the theory
of finite words. See \cite{starfreenote-meyer,cohen1993expressive}. 

Motivated by these applications, we look
for an analogue of the above proposition for the setting of traces.
We now point to some problems that arise if one tries to naively lift the Krohn-Rhodes theorem
to the setting of traces. 
The first problem is that, unlike in the word scenario, division
does not imply simulation of morphisms from the \emph{trace monoid}.
By simulation of morphisms from the trace monoid, we simply mean an obvious
adaptation of the Definition~\ref{wordsimulation} to the morphisms from the trace monoid.
\begin{example}[Example for Problem~$1$ of lifting Krohn Rhodes theorem to trace monoid]
	\begin{figure}[ht]
	\centering
	\resizebox{.6\textwidth}{!}{\begin{tikzpicture}
[->,>=stealth',shorten >=1pt,auto,node distance=2cm,
                    semithick, initial text=]
\tikzstyle{every state}=[draw=black]
  \node[initial,state] (1)                    {$q_1$};
  \node[state]         (a) [above right of=1] {$q_a$};
  \node[state]         (b) [below right of=1] {$q_b$};
  \node[state]         (ab) [below right of=a] {$q_{ab}$};
%  \node[state]         (E) [below of=D]       {$q_e$};
\begin{scope}[node distance=0.8cm]
\node [below of=b] {\text{Automata }A};
\end{scope}

  \path (1) edge              node {a} (a)
            edge              node [swap] {b} (b)
        (a) edge              node {b} (ab)
        (b) edge              node [swap]  {a} (ab);

  \node[initial,state] (1') [right of=ab]       {$q'_1$};
  \node[state]         (a') [above right of=1'] {$q'_a$};
  \node[state]         (b') [below right of=1'] {$q'_b$};
  \node[state]         (ab') [right of=a'] {$q'_{ab}$};
  \node[state]         (ba') [right of=b']       {$q'_{ba}$};
\begin{scope}[node distance=0.8cm]
\node [below of=b'] {\text{Automata }B};
\end{scope}
  \path (1') edge              node {a} (a')
            edge              node [swap] {b} (b')
        (a') edge              node {b} (ab')
        (b') edge              node [swap]  {a} (ba');
\end{tikzpicture}}
	\caption{Automata $A$ and $B$ on the alphabet $\{a,b\}$}%
\label{fig:pb1}%
	\end{figure}
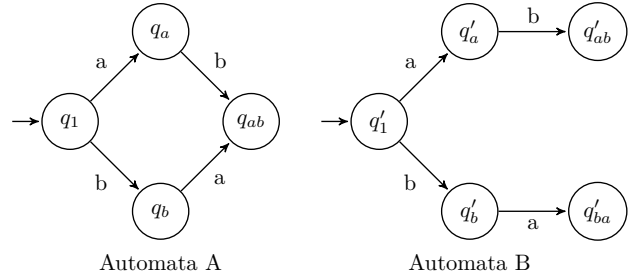
Consider the transition tm $(X,M)$ (resp.\ $(Y,N)$) of the automata $A$ (resp.\ $B$)
in Figure~\ref{fig:pb1};
assume both are complete, with any transition not shown in the figure going to trap state $q_t$ in $A$ and $q'_t$ in $B$.
So $X = \{q_1, q_a, q_b, q_{ab}, q_t\}$ and $M = \{1_M, m_a, m_b, m, 0_M\}$ where $1_M$ is the identity transformation of the 
empty word, and
$m_a, m_b$ and $m$ respectively represent the state transformations by $a$, $b$ and $ab$ (or equivalently, $ba$). 
$0_M$ represents the transformation by any other word. Hence the multiplication table of $M$ is the left one in Table~\ref{tab:MN-table}.
Similarly, $Y = \{q'_1, q'_a, q'_b, q'_{ab}, q'_{ba}, q'_t\}$ and $N = \{1_N, n_a, n_b, n_{ab}, n_{ba}, 0_N\}$ with its
multiplication given by the right one in Table~\ref{tab:MN-table}.
\begin{table}[ht]
\caption{Multiplication table of $M$ and $N$}%
\label{tab:MN-table}
\begin{tabular}{l|lllll}
            $M$ & $1_M$  & $m_a$ & $m_b$ & $m$ & $0_M$ \\ \hline
	$1_M$ & $1_M$ & $m_a$ & $m_b$ & $m$ & $0_M$ \\
	$m_a$ & $m_a$ & $0_M$ & $m$ & $0_M$ & $0_M$ \\
 $m_b$ & $m_b$ & $m$  & $0_M$ & $0_M$ & $0_M$\\
 $m$ & $m$ & $0_M$  & $0_M$ & $0_M$ & $0_M$ \\
 $0_M$ & $0_M$ & $0_M$ & $0_M$ & $0_M$ & $0_M$
\end{tabular}
\hspace{35pt}
\begin{tabular}{l|llllll}
	    $N$  & $1_N$  & $n_a$ & $n_b$ & $n_{ab}$ & $n_{ba}$ & $0_N$ \\ \hline
	$1_N$ &$1_N$  & $n_a$ & $n_b$ & $n_{ab}$ & $n_{ba}$ & $0_N$ \\
	$n_a$ & $n_a$ & $0_N$ & $n_{ab}$ & $0_N$ & $0_N$ & $0_N$\\
	$n_b$ & $n_b$ & $n_{ba}$  & $0_N$ & $0_N$ & $0_N$ & $0_N$ \\
	$n_{ab}$ & $n_{ab}$ & $0_N$  & $0_N$ & $0_N$ & $0_N$ & $0_N$ \\
	$n_{ba}$ & $n_{ba}$ & $0_N$  & $0_N$ & $0_N$ & $0_N$ & $0_N$ \\
	$0_N$ & $0_N$ & $0_N$ & $0_N$ & $0_N$ & $0_N$ & $0_N$
\end{tabular}
\end{table}
Observe that $(X,M) \divides (Y,N)$ by the pair $(f,\psi)$ where
	$\psi(n_a) = m_a$ and $\psi(n_b) = m_b$ extends to a surjective monoid
	morphism from $N$ to $M$. The surjective function $f$ maps $q'_{ab}$ and
	$q'_{ba}$ both to $q_{ab}$. Remaining details of the function are obvious.
	In particular, both the tm's can recognize the language $L = \{ab, ba\}$.

	Now consider the distributed alphabet $\tilde{\Sigma}=(\Sigma_1\!=\!\{a\}, \Sigma_2\!=\!\{b\})$. Clearly $a \ind b$,
	and $L$ is a trace language.
	Consider the function $\phi\colon  \alphabet \to M$ where $\phi(a) = m_a$ and $\phi(b) = m_b$. 
	As $m_a$ and $m_b$ commute, $\phi$ indeed extends to a morphism from the trace monoid, and can recognize $L$,
	for example.
	However, there is no morphism
	from $\traces$ to $(Y,N)$ that simulates $\phi$.
	Note that the `lifts' $n_a$ and $n_b$, of $m_a$ and $m_b$ resp.,  
	don't commute, and so the function that extends to 
	a simulating morphism $\psi$ in the word case as in the proof of Proposition~\ref{wordkr}, does not work here for traces.
	 \end{example}
The second problem is that even if there is a morphism from $\traces$ to a wreath 
product of tm's, in general it does not induce morphisms from trace monoids to the individual tm's
beyond the first one. This is primarily because the output of the \emph{sequential} transducer 
associated with the first tm is \emph{not} a trace.
\begin{example} 
	Assume the DFA in Figure~\ref{fig:pb2} represents the 
induced morphism to the first tm in a wreath product chain. 
The figure below shows the outputs of the sequential transducer associated
with this DFA on three different linearizations of a single input trace.
These outputs have different sets of letters and can not constitute a single trace.

\end{example}
	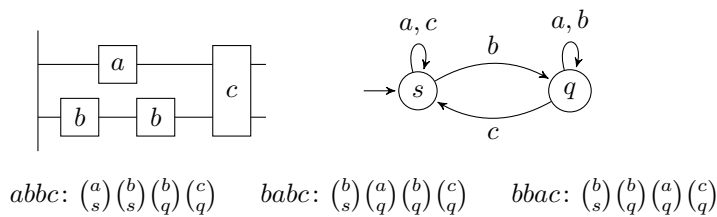
\begin{figure}[ht]%
		\centering
		\begin{tikzpicture}
\draw (0,0) -- (0, 3*0.2 + 2*0.5);
%\node at (-0.2, 0.2 + 0.5/2) {$2$};
%\node at (-0.2, 2*0.2 + 0.5 + 0.5/2) {$1$};

\draw (0.3,0.2) rectangle (0.8,0.7);
\draw (1.3,0.2) rectangle (1.8,0.7);
\draw (0.8,0.9) rectangle (1.3,1.4);
\draw (2.3,0.2) rectangle (2.8,1.4);

\draw (0,0.2 + 0.5/2) -- (0.3, 0.2 +0.5/2); 
\draw (0.8,0.2 + 0.5/2) -- (1.3,0.2 + 0.5/2); 
\draw (1.8,0.2 + 0.5/2) -- (2.3,0.2 + 0.5/2); 
\draw (2.8,0.2 + 0.5/2) -- (3.0,0.2 + 0.5/2);

\draw (0,2*0.2+0.5+0.5/2) -- (0.8,2*0.2+0.5+0.5/2); 
\draw (1.3,2*0.2+0.5+0.5/2) -- (2.3,2*0.2+0.5+0.5/2); 
\draw (2.8,2*0.2+0.5+0.5/2) -- (3.0,2*0.2+0.5+0.5/2); 

\node at (0.3 + 0.5/2, 0.2 + 0.5/2) {$b$};
\node at (0.3 + 1.0 + 0.5/2, 0.2 + 0.5/2) {$b$};
\node at (0.3 + 2.0 +  0.5/2, 0.2 + 0.5 + 0.2/2) {$c$};
\node at (0.3 + 0.5 + 0.5/2, 2*0.2 + 0.5 + 0.5/2) {$a$};

\begin{scope}
[->,>=stealth',shorten >=1pt,auto,node distance=2cm,
                    initial text=]
\tikzstyle{every state}=[draw=black, inner sep = 1mm, minimum size=0]

  \node[initial,state] at (5.0,0.8)  (1)                    {$s$};
  \node[state]                       (2) [right of=1]       {$q$};

\path (1) edge [bend left] node {$b$} (2)
          edge [loop above] node {$a,c$} (1)
      (2) edge [bend left] node {$c$} (1)
          edge [loop above] node {$a,b$} (2);

\end{scope}

\node at (0.5 + 0.5, -0.6) (l)  {$abbc \colon {\binom{a}{s}}
{\binom{b}{s}}{\binom{b}{q}}{\binom{c}{q}}$};
\node (m) [right of=l, xshift=2.3cm]    {$babc \colon {\binom{b}{s}}
{\binom{a}{q}}{\binom{b}{q}}{\binom{c}{q}}$};
\node (r) [right of=m, xshift=2.3cm]    {$bbac \colon {\binom{b}{s}}
{\binom{b}{q}}{\binom{a}{q}}{\binom{c}{q}}$};

\end{tikzpicture}
%	\resizebox{.45\textwidth}{!}{\input{figures/problem2.tikz}}
	%\includegraphics[height=12em]{pa2.pdf}
	\caption{Sequential transducer outputs for all linearizations of a trace}%
        \label{fig:pb2}
\end{figure}%
Prior work in~\cite{GUAIANA1998277} devised a wreath product principle for traces, but it 
uses non-trace structures to circumvent the second problem, thus limiting its applicability.

As seen in the previous sections, the new algebraic framework of asynchronous structures
supports true concurrency and is well suited to reason about trace languages. Most importantly,
an asynchronous morphism to a wreath product chain gives rise to asynchronous morphisms to
individual atm's of the chain (see the proof of Theorem~\ref{thm:wpp2} for an illustration).
This can be seen as a resolution of the second problem mentioned
above. 

Going ahead, we extend the notion of simulation to the case when 
the `simulating' morphism is an asynchronous morphism to an atm.

\begin{definition}
Let $\phi\colon  \traces \to T=(X,M)$ be a morphism to a tm.
Further, let $T'=(\{S_i\}, N)$ be an atm and $\psi\colon  \traces \to T'$ be an asynchronous morphism.
We say that $\psi$ is an asynchronous simulation of $\phi$ (or simply $\psi$ simulates $\phi$) if there exists
a surjective function $f\colon  S_\pset \to X$ such that, for all $a \in \Sigma$ and all $s \in \gs$,
$f(\psi(a)(s)) = \phi(a)(f(s))$.
\end{definition}

The fundamental theorem of Zielonka~\cite{zielonka1987notes} states that every recognizable language is accepted by 
some asynchronous automata. See~\cite{mukund2012automata} for another proof of the theorem. 
From the viewpoint of our algebraic setup and the previous definition, it 
guarantees the existence of a simulating asynchronous morphism.

\begin{theorem}[Zielonka Theorem]\label{zielonka}
Let $\phi\colon  \traces \to T$ be a morphism to a finite tm. 
There exists an asynchronous morphism $\psi\colon  \traces \to T'$, to a finite atm,
which simulates $\phi$.
\end{theorem}

Recall that the atm $U_2[\ell]$, defined in Example~\ref{ex:u2l}, is a natural extension of
the tm $U_2$ to the process $\ell$. In a similar vein, for a group $G$,
the atm $G[\ell]$ denotes the natural extension of the tm $(G,G)$ from Example~\ref{ex:g}
to the process $\ell$. We will use a similar notation to extend a tm to an atm localized
to a particular process.

Now we formulate the following decomposition question:

\begin{question}\label{qs-kr} 
	Let $\phi\colon \traces \to (X,M)$ be a morphism to a finite tm.
Does there exist an asynchronous morphism $\psi \colon \traces \to T'$ to a finite atm,
        such that $\psi$ simulates $\phi$, and  
the atm $T'$ is of 
the form $T_1 \wr \ldots \wr T_n$
where each factor $T_i$ is, for some $\ell \in \pset$, either the atm $U_2[\ell]$ or
is of the form $G[\ell]$ for some non-trivial simple group
$G$ dividing $M$ ?
\end{question}
	
In view of our discussion so far, it is clear that the above question asks for
a simultaneous generalization of the Krohn-Rhodes theorem for the setting of words (that is, Proposition~\ref{wordkr}), and
the Zielonka theorem for the setting of traces (that is, Theorem~\ref{zielonka}).
Question~\ref{qs-kr} in general remains open. However we answer it positively in a
particular case, namely that of acyclic architectures, which is general enough to include the common client-server settings.

\begin{definition}
	Let $\dalphabet = {\{\alphabet_i\}}_{i \in \pset}$ be
	a distributed alphabet. Then its communication graph 
	is $G = (\pset,E)$ where
  $E = \{ (i,j) \in \pset \times \pset \mid i \neq j \mbox{~and~} \alphabet_i
  \cap \alphabet_j \neq \emptyset \}$.  If the communication graph is acyclic,
  then the distributed alphabet is called an acyclic architecture.
\end{definition}
Observe that if $\dalphabet$ is an acyclic architecture, then no action
is shared by more than two processes. The work~\cite{acycliczielonka-anca} provides a simpler proof
of Zielonka's theorem in this case.

\begin{theorem}\label{thm:acyclickr}%
	If $\dalphabet$ is an acyclic architecture, then Question~\ref{qs-kr} admits a positive answer.
\end{theorem}
\begin{proof}
	The proof is by induction on the number of 
	processes. The base case with only one process follows from Proposition~\ref{wordkr}.
	
	For the general case, let $\pset = \{1, 2, \ldots, k \}$.
	Since the communication graph is acyclic, there exists a 
	`leaf' process which communicates with at most one other process.
	Without loss of generality, let the leaf process
	be $1$, and its only neighbouring process be $2$ (if process $1$ has no neighbour, then process $2$ can be any other process). 
	We `split' the given morphism $\varphi \colon \traces \to (X,M)$ based on the chosen leaf process~$1$.
	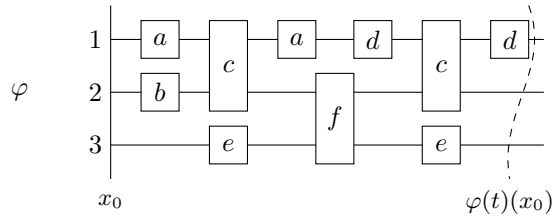
\begin{figure}[ht]
		\centering
		\begin{tikzpicture}
%event width = 0.5  , vertical space between events = 0.2 , horizontal space between events = 0.4
%processes 3
\draw (0,0) -- (0, 4 * 0.2 + 3 * 0.5);

\draw (0.4, 0.5 + 2*0.2) rectangle (0.4 + 0.5, 2*0.5 + 2*0.2);
\draw (0.4, 2*0.5 + 3*0.2) rectangle (0.4 + 0.5, 3*0.5 + 3*0.2);

\draw (2*0.4+0.5, 0.5 + 2*0.2) rectangle (2*0.4 + 2*0.5, 3*0.5 + 3*0.2);
\draw (2*0.4+0.5, 0.2) rectangle (2*0.4 + 2*0.5, 0.2 + 0.5);

\draw (3*0.4+2*0.5, 3*0.2 + 2*0.5) rectangle (3*0.4 + 3*0.5, 3*0.2 + 3*0.5);
\draw (3*0.4+3*0.5, 0.2 ) rectangle (3*0.4 + 4*0.5, 2*0.2 + 2*0.5);
\draw (3*0.4+4*0.5, 3*0.2 + 2*0.5 ) rectangle (3*0.4 + 5*0.5, 3*0.2 + 3*0.5);

\draw (4*0.4+5*0.5, 0.2 ) rectangle (4*0.4 + 6*0.5, 0.2 + 0.5);
\draw (4*0.4+5*0.5, 2*0.2 + 0.5 ) rectangle (4*0.4 + 6*0.5, 3*0.2 + 3*0.5);

\draw (5*0.4+6*0.5, 3*0.2 + 2*0.5 ) rectangle (5*0.4 + 7*0.5, 3*0.2 + 3*0.5);

%process 3
\draw (0, 0.2 + 0.5/2) -- (2*0.4 + 0.5, 0.2 + 0.5/2);
\node at (-0.2, 0.2 + 0.5/2) {$3$};
\draw (2*0.4 + 2*0.5, 0.2 + 0.5/2) -- (3*0.4 + 3*0.5, 0.2 + 0.5/2);
\draw (3*0.4 + 4*0.5, 0.2 + 0.5/2) -- (4*0.4 + 5*0.5, 0.2 + 0.5/2);
\draw (4*0.4 + 6*0.5, 0.2 + 0.5/2) -- (5*0.4 + 7*0.5 + 0.2, 0.2 + 0.5/2);

%process2
\draw (0, 2*0.2 + 0.5 + 0.5/2) -- (0.4, 2*0.2 + 0.5 + 0.5/2);
\node at (-0.2, 2*0.2 + 0.5 + 0.5/2) {$2$};
\draw (0.4 +0.5, 2*0.2 + 0.5 + 0.5/2) -- (2*0.4 + 0.5, 2*0.2 + 0.5 + 0.5/2);
\draw (2*0.4 + 2*0.5, 2*0.2 + 0.5 + 0.5/2) -- (3*0.4 + 3*0.5, 2*0.2 + 0.5 + 0.5/2);
\draw (3*0.4 + 4*0.5, 2*0.2 + 0.5 + 0.5/2) -- (4*0.4 + 5*0.5, 2*0.2 + 0.5 + 0.5/2);
\draw (4*0.4 + 6*0.5, 2*0.2 + 0.5 + 0.5/2) -- (5*0.4 + 7*0.5 + 0.2, 2*0.2 + 0.5 + 0.5/2);

%process2
\draw (0, 3*0.2 + 2*0.5 + 0.5/2) -- (0.4, 3*0.2 + 2*0.5 + 0.5/2);
\node at (-0.2, 3*0.2 + 2*0.5 + 0.5/2) {$1$} ;
\draw (0.4 +0.5, 3*0.2 + 2*0.5 + 0.5/2) -- (2*0.4 + 0.5, 3*0.2 + 2*0.5 + 0.5/2);
\draw (2*0.4 + 2*0.5, 3*0.2 + 2*0.5 + 0.5/2) -- (3*0.4 + 2*0.5,3*0.2 + 2*0.5 + 0.5/2);
\draw (3*0.4 + 3*0.5, 3*0.2 + 2*0.5 + 0.5/2) -- (3*0.4 + 4*0.5,3*0.2 + 2*0.5 + 0.5/2);
\draw (3*0.4 + 5*0.5, 3*0.2 + 2*0.5 + 0.5/2) -- (4*0.4 + 5*0.5,3*0.2 + 2*0.5 + 0.5/2);
\draw (4*0.4 + 6*0.5, 3*0.2 + 2*0.5 + 0.5/2) -- (5*0.4 + 6*0.5,3*0.2 + 2*0.5 + 0.5/2);
\draw (5*0.4 + 7*0.5, 3*0.2 + 2*0.5 + 0.5/2) -- (5*0.4 + 7*0.5 + 0.2,3*0.2 + 2*0.5 + 0.5/2);

%labels
\node at (0.4 + 0.5/2, 3*0.2 + 2*0.5 + 0.5/2) {$a$};
\node at (0.4 + 0.5/2, 2*0.2 + 0.5 + 0.5/2) {$b$};

\node at (2*0.4 + 0.5 + 0.5/2, 2*0.2 + 2*0.5 + 0.2/2) {$c$};
\node at (2*0.4 + 0.5 + 0.5/2, 0.2 + 0.5/2) {$e$};

\node at (3*0.4 + 2*0.5 +0.5/2, 3*0.2 + 2*0.5 + 0.5/2) {$a$};

\node at (3*0.4 + 3*0.5 +0.5/2, 0.2 + 0.5 + 0.2/2) {$f$};

\node at (3*0.4 + 4*0.5 +0.5/2, 3*0.2 + 2*0.5 + 0.5/2) {$d$};
\node at (4*0.4 + 5*0.5 + 0.5/2, 2*0.2 + 2*0.5 + 0.2/2) {$c$};
\node at (4*0.4 + 5*0.5 + 0.5/2, 0.2 + 0.5/2) {$e$};

\node at (5*0.4 + 6*0.5 +0.5/2, 3*0.2 + 2*0.5 + 0.5/2) {$d$};

%configs 
\node at (0,-0.3) {\small ${x_0}$};
%\node at (0, 4*0.2 + 3*0.5 +0.3) {\small $\mbox{id}$};

\draw[dashed] (7*0.5 + 5*0.4, 3*0.5 + 4*0.2) to [out=-70, in = 100]  (6*0.5 + 5*0.4 + 0.5/2, 0);

\node at (6*0.5 + 5*0.4 + 0.5/2,-0.3) {\small $\phi(t)(x_0)$};
%\node at (7*0.5 + 5*0.4, 4*0.2 + 3*0.5 +0.3) {\small $\phi(t)$};

\node at (-1.2, 2*0.2+ 0.5 +0.5/2) {$\varphi$};

\end{tikzpicture}
		\caption{Initial and final states of $(X,M)$ under $\varphi$}
		\label{fig:tree-varphi-over-t}
	\end{figure}

	\begin{figure}[ht]
\centering
\begin{tikzpicture}
%event width = 0.5  , vertical space between events = 0.2 , horizontal space between events = 0.4
%processes 3
\draw (0,0) -- (0, 3 * 0.2 + 2 * 0.5);
			
\draw (0,6*0.2 +3*0.5) -- (0, 8 * 0.2 + 4 * 0.5);

\draw (0.4, 0.5 + 2*0.2) rectangle (0.4 + 0.5, 2*0.5 + 2*0.2);
\draw (0.4, 7*0.2 +3*0.5) rectangle (0.4 + 0.5, 4*0.5 + 7*0.2); %new only p1

\draw (2*0.4+0.5 - 0.1, 0.5 + 2*0.2) rectangle (2*0.4 + 2*0.5, 2*0.5 + 3*0.2);
\draw (2*0.4+0.5, 0.2) rectangle (2*0.4 + 2*0.5, 0.2 + 0.5);

%\draw (3*0.4+2*0.5, 3*0.2 + 2*0.5) rectangle (3*0.4 + 3*0.5, 3*0.2 + 3*0.5);
\draw (3*0.4+3*0.5, 0.2 ) rectangle (3*0.4 + 4*0.5, 2*0.2 + 2*0.5);
	%\draw (3*0.4+4*0.5, 3*0.2 + 2*0.5 ) rectangle (3*0.4 + 5*0.5, 3*0.2 + 3*0.5);

\draw (4*0.4+5*0.5, 0.2 ) rectangle (4*0.4 + 6*0.5, 0.2 + 0.5);
\draw (4*0.4+5*0.5 - 0.3, 2*0.2 + 0.5 ) rectangle (4*0.4 + 6*0.5, 2*0.5 + 3*0.2);

	%\draw (5*0.4+6*0.5, 3*0.2 + 2*0.5 ) rectangle (5*0.4 + 7*0.5, 3*0.2 + 3*0.5);

			%process 3
\draw (0, 0.2 + 0.5/2) -- (2*0.4 + 0.5, 0.2 + 0.5/2);
\node at (-0.2, 0.2 + 0.5/2) {$3$} ;
\draw (2*0.4 + 2*0.5, 0.2 + 0.5/2) -- (3*0.4 + 3*0.5, 0.2 + 0.5/2);
\draw (3*0.4 + 4*0.5, 0.2 + 0.5/2) -- (4*0.4 + 5*0.5, 0.2 + 0.5/2);
\draw (4*0.4 + 6*0.5, 0.2 + 0.5/2) -- (5*0.4 + 7*0.5 + 0.2, 0.2 + 0.5/2);

			%process2
\draw (0, 2*0.2 + 0.5 + 0.5/2) -- (0.4, 2*0.2 + 0.5 + 0.5/2);
\node at (-0.2,2*0.2 + 0.5 + 0.5/2) {$2$} ;
\draw (0.4 +0.5, 2*0.2 + 0.5 + 0.5/2) -- (2*0.4 + 0.5 -0.1, 2*0.2 + 0.5 + 0.5/2);
\draw (2*0.4 + 2*0.5, 2*0.2 + 0.5 + 0.5/2) -- (3*0.4 + 3*0.5, 2*0.2 + 0.5 + 0.5/2);
\draw (3*0.4 + 4*0.5, 2*0.2 + 0.5 + 0.5/2) -- (4*0.4 + 5*0.5 -0.3, 2*0.2 + 0.5 + 0.5/2);
\draw (4*0.4 + 6*0.5, 2*0.2 + 0.5 + 0.5/2) -- (5*0.4 + 7*0.5 + 0.2, 2*0.2 + 0.5 + 0.5/2);

			%process1
\draw (0, 7*0.2 +3*0.5 + 0.5/2 ) -- (0.4, 7*0.2 + 3*0.5 + 0.5/2);
\node at (-0.2,7*0.2 +3*0.5 + 0.5/2) {$1$} ;
\draw (0.4 +0.5, 7*0.2 +3*0.5 + 0.5/2) -- (2*0.4 + 0.5,7*0.2 +3*0.5 + 0.5/2);
\draw (2*0.4+0.5, 7*0.2 +3*0.5) rectangle (2*0.4 + 2*0.5, 4*0.5 + 7*0.2); %new only p1
\draw (2*0.4 + 2*0.5, 7*0.2 +3*0.5 + 0.5/2) -- (3*0.4 + 2*0.5,7*0.2 +3*0.5 + 0.5/2);
\draw (3*0.4 + 2*0.5, 7*0.2 +3*0.5) rectangle (3*0.4 + 3*0.5, 4*0.5 + 7*0.2); %new only p1
\draw (3*0.4 + 3*0.5, 7*0.2 +3*0.5 + 0.5/2) -- (3*0.4 + 4*0.5,7*0.2 +3*0.5 + 0.5/2);
\draw (3*0.4 + 4*0.5, 7*0.2 +3*0.5) rectangle (3*0.4 + 5*0.5, 4*0.5 + 7*0.2); %new only p1
\draw (3*0.4 + 5*0.5, 7*0.2 +3*0.5 + 0.5/2) -- (4*0.4 + 5*0.5,7*0.2 +3*0.5 + 0.5/2);
\draw (4*0.4 + 5*0.5, 7*0.2 +3*0.5) rectangle (4*0.4 + 6*0.5, 4*0.5 + 7*0.2); %new only p1
\draw (4*0.4 + 6*0.5, 7*0.2 +3*0.5 + 0.5/2) -- (5*0.4 + 6*0.5,7*0.2 +3*0.5 + 0.5/2);
\draw (5*0.4 + 6*0.5, 7*0.2 +3*0.5) rectangle (5*0.4 + 7*0.5, 4*0.5 + 7*0.2); %new only p1
\draw (5*0.4 + 7*0.5, 7*0.2 +3*0.5 + 0.5/2) -- (5*0.4 + 7*0.5 + 0.2,7*0.2 +3*0.5 + 0.5/2);

			%labels
\node at (0.4 + 0.5/2, 7*0.2 +3*0.5 + 0.5/2) {$a$};
\node at (0.4 + 0.5/2, 2*0.2 + 0.5 + 0.5/2) {$b$};
\node at (2*0.4 + 0.5 + 0.5/2 -0.04, 2*0.2 + 0.5 +0.17) {\small $\varphi(a)$};

\node at (2*0.4 + 0.5 + 0.5/2 -0.05, 2*0.2 + 0.5 + 0.5) {$c$};
\node at (2*0.4 + 0.5 + 0.5/2,7*0.2 +3*0.5 + 0.5/2 ) {$c$};
\node at (2*0.4 + 0.5 + 0.5/2, 0.2 + 0.5/2) {$e$};

\node at (3*0.4 + 2*0.5 +0.5/2, 7*0.2 +3*0.5 + 0.5/2) {$a$};

\node at (3*0.4 + 3*0.5 +0.5/2, 0.2 + 0.5 + 0.2/2) {$f$};

\node at (3*0.4 + 4*0.5 +0.5/2, 7*0.2 +3*0.5 + 0.5/2) {$d$};
\node at (4*0.4 + 5*0.5 + 0.5/2 -0.15,2*0.2 + 0.5 + 0.5 ) {$c$};
\node at (4*0.4 + 5*0.5 + 0.5/2,7*0.2 +3*0.5 + 0.5/2 ) {$c$};
\node at (4*0.4 + 5*0.5 + 0.5/2 -0.15, 2*0.2 + 0.5 +0.17) {\small $\varphi(ad)$};
\node at (4*0.4 + 5*0.5 + 0.5/2, 0.2 + 0.5/2) {$e$};

\node at (5*0.4 + 6*0.5 +0.5/2,7*0.2 +3*0.5 + 0.5/2 ) {$d$};

			%configs 
\node at (0,-0.3) {\small $x_0$};
%\node at (0, 3*0.2 + 2*0.5 +0.3) {\small $\mbox{id}$};

\draw[dashed] (7*0.5 + 5*0.4,8*0.2 +4*0.5) to [out=-70, in = 70]  (7*0.5 + 5*0.4, 3*0.5 + 6*0.2 );
\draw[dashed] (0.5 + 0.4 +0.1, 8*0.2 +4*0.5) to [out=-70, in = 60]  (0.5 + 0.4, 3*0.5 + 6*0.2 );
\draw[dashed] (5*0.5 + 3*0.4 +0.1, 8*0.2 +4*0.5) to [out=-70, in = 70]  (5*0.5 + 3*0.4, 3*0.5 + 6*0.2 );

%\draw[dashed] (7*0.5 + 5*0.4, 3*0.5 + 4*0.2) to [out=-70, in = 100]  (6*0.5 + 5*0.4 + 0.5/2, 0);
%\draw[dashed] (0.5 + 0.4 +0.1, 3*0.5 + 4*0.2) to [out=-70, in = 80]  (0.5 + 0.4, 0);
%\draw[dashed] (5*0.5 + 3*0.4 +0.1, 3*0.5 + 4*0.2) to [out=-70, in = 100]  (4*0.5 + 3*0.4 + 0.5/2, 0);

\node at (0.5 + 0.4 -0.1 ,5*0.2 + 3*0.5) {\small $\phi(a)$};
%\node at (0.5 + 0.4 +0.1, 3*0.5 + 4*0.2 +0.3) {\small $\phi(a)$};

\node at (4*0.5 + 3*0.4 + 0.5/2 -0.1 , 5*0.2 + 3*0.5) {\small $\phi(ad)$};
%\node at (5*0.5 + 3*0.4 +0.1, 3*0.5 + 4*0.2 +0.3) {\small $\overline{\phi(ad)}$};

\node at (6*0.5 + 5*0.4 + 0.5/2 , 5*0.2 + 3*0.5) {\small $\phi(d)$};
%\node at (7*0.5 + 5*0.4, 4*0.2 + 3*0.5 +0.3) {\small $\overline{\phi(d)}$};
		%configs 
\node at (0, 4*0.2 + 3*0.5 +0.2) {\small $\mbox{id}$};
%\node at (0, 8 * 0.2 + 4 * 0.5 + 0.2) {\small $\mbox{id}$};

\draw[dashed] (7*0.5 + 5*0.4, 2*0.5 + 2*0.2) to [out=-70, in = 100]  (6*0.5 + 5*0.4 + 0.5/2, -0.1);
\node at (6*0.5 +5*0.4 +0.5/2, -0.3) {\small $x$};

\node at (-1.3, 0.8) {\small $\varphi_2$};
\node at (-1.3,7*0.2 +3*0.5 + 0.5/2) {\small $\varphi_1$};

%\draw[dashed] (2*0.5 + 2*0.4 +0.1, 3*0.5 + 3*0.2) to [out=-85, in = 5]  (2*0.5 + 2* 0.4, 0.2 + 0.5 + 0.2/2);
%\draw[dashed] (2*0.5 + 2*0.4 , 0.2 + 0.5 + 0.2/2) to [out=190, in = 0]  (0.5 + 2*0.4, 0.2 + 0.5 + 0.2/2);
%\draw[dashed] (0.5 + 2*0.4 , 0.2 + 0.5 + 0.2/2) to [out=180, in = 70]  (0.5 + 0.4, 0);
%\draw[dashed] (6*0.5 + 4*0.4 +0.15, 3*0.5 + 3*0.2) to [out=-85, in = 5]  (6*0.5 + 4* 0.4, 0.2 + 0.5 + 0.2/2);
%\draw[dashed] (6*0.5 + 4*0.4 , 0.2 + 0.5 + 0.2/2) to [out=190, in = 0]  (5*0.5 + 4*0.4, 0.2 + 0.5 + 0.2/2);
%\draw[dashed] (5*0.5 + 4*0.4 , 0.2 + 0.5 + 0.2/2) to [out=180, in = 70]  (5*0.5 + 3*0.4, 0);

%arrows
\draw[thick,->] (0.5 + 0.4 -0.1 ,4*0.2 + 3*0.5) to [out=-90, in=120] (2*0.4+0.5 - 0.17, 2*0.5 + 2*0.2); 
\draw[thick,->] (4*0.5 + 3*0.4 + 0.5/2 -0.1 ,4*0.2 + 3*0.5) to [out=-90, in=120] (4*0.4+5*0.5 - 0.37, 2*0.5 + 2*0.2); 

%dotted line between varphi1 and varphi2
\draw[dotted, thick] (-0.3, 3*0.2 + 2*0.5 + 0.3) -- (5*0.4 + 6*0.5 +0.5 +0.2, 3*0.2 + 2*0.5 + 0.3);

\end{tikzpicture}
\caption{Transfer of state information from $\varphi_1$ to $\varphi_2$. The final states of the two atm's are $\varphi(d)$ and $x$.
Note that $\varphi(d)(x) = \varphi(t)(x_0)$.}
\label{fig:treeproofsec}
\end{figure}
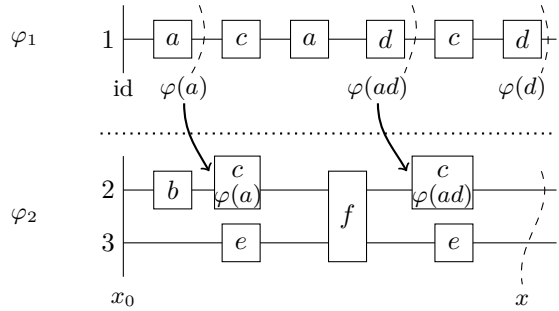
	\paragraph*{Defining $\varphi_1$ and $\varphi_2$}
	Let $N$ be the submonoid of $M$ generated by 
	$\{\phi(a) \mid \loc(a) = \{1\}\}$. Also let 
	$\overline{N}$ be the semigroup of \emph{reset} (that is, constant)
	functions from $N$ into itself. If $n \in N$, we denote by $\overline{n}$
	the function on $N$ which maps every element to $n$.

	We define $\varphi_1 \colon  \alphabet_1^* \to (N,N \cup \overline{N})$ by setting 
	\begin{align*}
		\varphi_1(a) &= \varphi(a) &&\text{ if } \loc(a) = \{1\} \\
		\varphi_1(a) &= \overline{\mbox{id}}  &&\text{ if } \loc(a) = \{1,2\}
	\end{align*}
        Note that at any point, $\varphi_1$ records in the state of the tm, 
        the evaluation $\varphi(w)$ where $w \in (\Sigma_1 \setminus \Sigma_2)^*$ is
	the word read by process~$1$ since the last joint action with its neighbour.
        As a result, the sequential transducer 
	associated with $\varphi_1$ adds the corresponding information
	at each process~$1$ event.

	In particular, the information supplied at the joint events of process~$1$ and $2$, will be used
	by $\varphi_2$. For this,
	let us define a suitable distributed alphabet $\wt{\alphabet'} = \{\alphabet'_2, 
	\alphabet_3, \ldots, \alphabet_k \}$ over $\pset \setminus \{1\}$,
	where $\alphabet_2' = (\alphabet_2 \setminus \alphabet_1) \cup ((\alphabet_1 
	\cap \alphabet_2) \times N)$.  We define 
	$\varphi_2 \colon  TR(\wt{\alphabet'}) \to (X,M)$ by letting
	\begin{align*}
		\varphi_2(a) &=  \phi(a) &&\text{ if } 1 \notin \loc(a) \\
		\varphi_2((a,n)) &= n\phi(a)  &&\text{ if } \loc(a) = \{1,2\}
	\end{align*}
	We denote the total alphabet corresponding to $\dalphabet$ and $\wt{\alphabet'}$ by $\alphabet$ and
	$\alphabet'$ respectively. For any two letters
	$a, b \in \alphabet \cap \alphabet' = \alphabet \setminus \alphabet_1$, if $a$ and $b$ 
	are independent in the new
	distributed alphabet $\wt{\alphabet'}$, then they must have been independent in $\dalphabet$. For independent letters $(a,n)$ and 
	$b$ in $\wt{\alphabet'}$, it is easy to show that, process~$1$ being a leaf process, $a$ and $b$ are independent in $\dalphabet$, and that
	$\varphi(b)$ and $n$ commute. Based on these, we can verify that $\varphi_2$ is indeed a morphism. 	
        \paragraph*{\textbf{Simulating $\varphi_1$ and $\varphi_2$}}
	By induction hypothesis, we get a simulating morphism of $\varphi_1$, namely 
	$\widehat{\varphi_1}:
	\alphabet_1^* \to T_1 \wr T_2 
	\wr \ldots \wr T_n$ where each factor is of the form $U_2$ or
	$(G,G)$ for some non-trivial simple group $G$ dividing $N \cup \overline{N}$. Let 
	$T_1 \wr T_2 \wr \ldots \wr T_n = T = (Y, M_1)$. Then, by 
	definition of simulation, there exists a surjective mapping $f_1 : Y \to N$ such
	that for any $y \in Y$ and any $a \in \alphabet_1$,  $f_1(\widehat{\varphi_1}(a)(y)) = \varphi_1(a)(f_1(y))$.

        We define a morphism $\psi_1 : \traces \to T[1]$ by setting
	\begin{align*}
		\psi_1(a) &= \widehat{\varphi_1}(a) &&\text{ if } a \in \alphabet_1 \\
		\psi_1(a) &= \mbox{id} &&\text{ otherwise }
	\end{align*}
        It is easy to check that $\psi_1$ is an asynchronous morphism. Also, it's not difficult
	to see that $T[1] = T_1[1] \wr \ldots \wr T_n[1]$. We write $T[1]$ as $(\{Y_i\}, M_1)$. 
	Since the process~$1$ local
	states represent the global states of $T[1]$, we can consider $f_1$ as a surjective function
	from $Y_\pset$ to $N$ such that for any $y \in Y_\pset$, and $a \in \alphabet_1$, 
	we have $f_1(\psi_1(a)(y)) = \varphi_1(a)(f_1(y))$.

%\begin{figure}[b]
%		\centering
%		\begin{tikzpicture}
%			\node  at (0,0) (sw) {$N$};
%			\node at (1.4,0) (se) {$N$};
%			\node  at (0,1.4) (nw) {$Y_\pset$};
%			\node at (1.4,1.4) (ne) {$Y_\pset$};
%			\draw[->] (sw.east) to node[below] {\small $\varphi_1(a)$} (se.west); 
%			\draw[->] (nw.east) to node[above] {\small $\psi_1(a)$}(ne.west); 
%			\draw[->] (nw.south) to node[left] {\small $f_1$} (sw.north); 
%			\draw[->] (ne.south) to node[right] {\small $f_1$} (se.north); 
%		
%			\node  at (4,0) (sw) {$X$};
%			\node at (5.4,0) (se) {$X$};
%			\node  at (4,1.4) (nw) {$Z_\pset$};
%			\node at (5.4,1.4) (ne) {$Z_\pset$};
%			\draw[->] (sw.east) to node[below] {\small $\varphi_2(a)$} (se.west); 
%			\draw[->] (nw.east) to node[above] {\small $\psi_2(a)$}(ne.west); 
%			\draw[->] (nw.south) to node[left] {\small $f_2$} (sw.north); 
%			\draw[->] (ne.south) to node[right] {\small $f_2$} (se.north); 
%		\end{tikzpicture}
%		\caption{$\psi_1$ simulates $\varphi_1$ for $a \in \alphabet_1$. For $a \notin \alphabet'_1 \setminus \alphabet'_2$,
%		the morphism $\psi_2$ simulates $\varphi_2$}
%		\label{fig:commute-tree-1st}
%	\end{figure}
        Note that, by construction, each
	$T_m[1]$ is of the form $U_2[1]$ or $G[1]$ for some non-trivial simple group $G$ dividing 
	$N \cup \overline{N}$. If $G \divides N \cup \overline{N}$, namely there exists a surjective 
	morphism $\tau$ from a submonoid $N'$ of $N \cup \overline{N}$ onto $G$, then
	$\tau(\overline{n}) = \mbox{id}_G$ for every $n \in N$. Indeed $\overline{n}$ is an idempotent,
	so its $\tau$-image must be the only 
	idempotent in $G$, namely $\mbox{id}_G$. Clearly $N'' = N' \cap N$ is a sub-monoid 
	of $N$ and $\tau' :  N'' \to G$ defined by $\tau'(n) = \tau(n)$ is a surjective morphism
	from $N''$ to $G$. So, $G \divides N$ and since division is
	transitive, $G \divides M$. 
        
	Similarly for $\varphi_2$, by the induction hypothesis, we have a simulating morphism 
	$\widehat{\varphi_2} \colon TR(\wt{\alphabet'}) \to T'$, where $T'  = T_1' \wr \ldots
	\wr  T'_{n'}$, with
        each factor of the form $U_2[\ell]$ or $G[\ell]$ for some simple group $G$ dividing $M$,
	and some $\ell \in \{2, \ldots ,k\}$. Similar to what we did previously, we tweak these
	atm's to make them work over $\pset$, by adding a singleton set of local states for process~$1$. 
	If we denote this by $T'_m[{\uparrow}1]$, then $T'[{\uparrow}1] = T'_1[{\uparrow}1] 
	\wr \ldots \wr T'_{n'}[{\uparrow}1]$. Consider the distributed alphabet 
        $\wt{\alphabet''} = (\alphabet'_1, \alphabet'_2, \alphabet_2, \ldots, \alphabet_k)$,
	where $\alphabet'_1 = (\alphabet_1 \setminus \alphabet_2) \cup ((\alphabet_1 \cap \alphabet_2) 
	\times N)$.
	We devise a morphism $\psi_2 \colon TR(\wt{\alphabet''}) \to T'[{\uparrow}1]$ by 
	setting
	\begin{align*}
	\psi_2(a) &= \mbox{id} && \text{ if } a \in \alphabet_1 \setminus \alphabet_2 \\
	\psi_2(a) &= \widehat{\varphi_2}(a) && \text{ otherwise }
	\end{align*}
	Let us denote $T'[{\uparrow}1]$ as the atm 
	$(\{Z_i\}, M_2)$. 
	Due to the canonical bijection between global states of $T'$ and $T'[{\uparrow}1]$, there is a
	surjective function $f_2 \colon Z_\pset \to X$ such that for any $a \in \alphabet \setminus
	\alphabet_1$, and any $z \in Z_\pset$, we have $f_2(\psi_2(a)(z)) = \varphi_2(a)(f_2(z))$.
	Furthermore, for $(a,n) \in (\alphabet_1 \cap \alphabet_2) \times N$,
	and any $z \in Z_\pset$, we have $f_2(\psi_2((a,n))(z)) = \varphi_2((a,n))(f_2(z))$.

	\paragraph*{\textbf{Asynchronously simulating $\phi$}}
	
	The final step is to combine $\psi_1$ and $\psi_2$
	to get asynchronous
	morphism $\psi: \traces \to (\{Y_i\}, M_1) \wr (\{Z_i\}, M_2)$ such that $\psi$
	simulates $\phi$. Recall that $(\{Y_i\}, M_1) \wr (\{Z_i\}, M_2) = 
	(\{Y_i \times Z_i\}, M_1 \times {\mathcal F}(Y_\pset, M_2))$. We define 
	$\psi$ as follows:
	\begin{equation*}\psi(a) = (\psi_1(a), \gamma_a) \text{ where } \gamma_a \colon Y_\pset \to M_2 \text{ is given by}\end{equation*}
	\begin{align*}
	& \gamma_a(y) = \mbox{id} & \text{if } \loc(a) = \{1\} \\
		& \gamma_a(y) = \psi_2((a, f_1(y)))& \text{if } \loc(a) = \{1,2\} \\
	& \gamma_a(y) = \psi_2(a) & \text{if } 1 \notin \loc(a)
	\end{align*}
	For $a \in \alphabet_1 \cap \alphabet_2$, recall that in the 
	first atm $T[1]$, any global state is completely determined by its process~$1$ state. So,
	$\psi(a)$ is an $a$-map, and hence $\psi$
	is an asynchronous morphism.
	\begin{figure}[h]
		\centering
		\begin{tikzpicture}
			\node  at (0,0) (sw) {$X$};
			\node at (2.7,0) (se) {$X$};
			\node  at (0,1.4) (nw) {$Y_\pset \times Z_\pset$};
			\node at (2.7,1.4) (ne) {$Y_\pset \times Z_\pset$};
			\draw[->] (sw.east) to node[below] {\small $\varphi(a)$} (se.west); 
			\draw[->] (nw.east) to node[above] {\small $\psi(a)$}(ne.west); 
			\draw[->] (nw.south) to node[left] {\small $f$} (sw.north); 
			\draw[->] (ne.south) to node[right] {\small $f$} (se.north); 
		\end{tikzpicture}
		\caption{$\psi$ simulates $\varphi$}
		\label{fig:commute-tree-final}
	\end{figure}

	We now show that there exists a surjective function $f : Y_\pset \times Z_\pset \to X$
	such that $f(\psi(a)(y,z)) = \phi(a)(f(y,z))$ for all $(y,z) \in Y_\pset \times Z_\pset$.
	We define $f(y,z) = f_1(y)(f_2(z))$. It is surjective because both $f_1$ and $f_2$ are
	surjective, and $N$ contains an identity element. Simple calculations show that $\psi$ simulates 
	$\varphi$. We give a case by case argument as to why
        this should be true. 

	$f(\psi(a)(y,z))$ refers to the new state of $(X,M)$ that we get
	by first reading the letter $a$ at state
	$(y,z)$ in the atm $(\{Y_i\}, M_1) \wr (\{Z_i\}, M_2)$, and then mapping back to the 
	corresponding state of $(X,M)$ using $f$.

	\underline{$\mathbf{\textbf{Case }a \in \alphabet_1 \cap \alphabet_2}$:}
	When $a \in \alphabet_1 \cap \alphabet_2$, it is a joint letter of process~$1$ and process~$2$.
	Recall from the definition of $\varphi_1$, that these joint letters reset the state of the tm 
	$(N, N \cup \overline{N})$ to the $\mbox{id}$ state. Because $\psi_1$ simulates $\varphi_1$ on letters from $\alphabet_1$,
	it should be clear that the new state $\psi_1(a)(y)$ of $(\{Y_i\}, M_1)$ maps to state $\mbox{id}$ in 
	the tm $(N, N \cup \overline{N})$. That is $f_1(\psi_1(a)(y)) = \mbox{id}$.

	The local asynchronous transducer of the first atm adds the process~$1$ state (equivalent to
	global state in the first atm) $y$ to the letter $a$. This coresponds to the letter $(a, f_1(y))$
	as input to $\varphi_2$. Before reading the letter, the tm $(X,M)$ (for $\varphi_2$) is in state
	$f_2(z)$. Thus, from the definition of $\varphi_2$, and because $\psi_2$ simulates $\varphi_2$ on the
	extended letters, we know that the new state $\psi_2(a, f_1(y))(f_2(z))$ in $(\{Z_i\}, M_2)$ should map to the state 
	$f_1(y).\varphi(a)(f_2(z)) = \varphi(a)(f(y,z))$ in $(X,M)$. That is, $f_2(\psi_2(a, f_1(y))(f_2(z))) 
	= \varphi(a)(f(y,z))$. The overall state of $(X,M)$ is then given by $\mbox{id}(\varphi(a)(f(y,z))) =
	\varphi(a)(f(y,z))$. Hence in this case $\psi$ simulates $\varphi$.

        \begin{align*}
	f(\psi(a)(y,z)) &= f(\psi_1(a)(y), \psi_2(a, f_1(y)(z))) \\
			&= f_1(\psi_1(a)(y))(f_2(\psi_2(a,f_1(y)(z)))) \\
			&= \varphi_1(a)f_1(y)(\varphi_2((a,f_1(y))(f_2(z)) \\
			&= \mbox{id}(f_1(y)\varphi(a)(f_2(z))) \\
			&= \varphi(a)(f(y,z))
	\end{align*}

        \underline{$\mathbf{\textbf{Case }a \in \alphabet_1 \setminus \alphabet_2}$:}
         In this case, state of $(X,M)$ (from $\varphi_2$) doesn't change, that is, it remains $f_2(z)$.
	 The state of $(N, N \cup \overline{N})$ should be $\varphi(a)$ applied to the old state $f_1(y)$.
	 That is, the new state of $(X,M)$ should be $\varphi(a)(f_1(y))(f_2(z))$. Note that $\varphi(a)(f_1(y))$
	 is the element $f_1(y).\varphi(a)$ in $N$. Thus new state is $f_1(y).\varphi(a)(f_2(z)) = \varphi(a)(f(y,z))$.
	 Hence, in this case $\psi$ simulates $\varphi$.

        \begin{align*}
		f(\psi(a)(y,z)) &= f(\psi_1(a)(y), \mbox{id}(z)) \\
			&= f_1(\psi_1(a)(y))(f_2(z)) \\
			&= \varphi_1(a)f_1(y)(f_2(z)) \\
			&= \varphi(a)(f(y,z))
	\end{align*}

        \underline{$\mathbf{\textbf{Case }a \notin \alphabet_1}$:}
	In this case, state of tm for $\varphi_1$ should not update. So its new state is the old state $f_1(y)$.
	The new state of $(X,M)$ (for $\varphi_2$) should be $\varphi(a)$ applied to old state $f_2(z)$. Thus, the new
	overall state of $(X,M)$ is $f_1(y)(\varphi(a)(f_2(z))) = \varphi(a).f_1(y)(f_2(z))$. Since $a \notin \alphabet_1$,
	note that $f_1(y) \in N$ commutes with $\varphi(a)$. Thus, new state is $f_1(y)\varphi(a)(f_2(z)) =
	\varphi(a)(f(y,z))$. hence in this case also, $\psi$ simulates $\varphi$.
	and $\varphi(a)$
        
	\begin{align*}
		f(\psi(a)(y,z)) &= f(\psi_1(a)(y), \psi_2(a)(z)) \\
				&= f_1(\psi_1(a)(y))(f_2(\psi_2(a)(z))) \\
				&= f_1(y)(\varphi_2(a)(f_2(z))) \\
				&= \varphi_2(a)f_1(y)(f_2(z)) \\
				&= f_1(y)\varphi_2(a)(f_2(z)) \\
				&= \varphi(a)(f(y,z))
	\end{align*}
\end{proof}
\section{Local and Global Cascade Products}\label{sec:cascade}
In this section, we introduce distributed automata-theoretic operations called local and global cascade products.

\subsection{Local Cascade Product}

\begin{definition}\label{def:localcascade}
Let $\aaa_1 = (\{\ls_i\}, \{\lt_a\}, s_{\text{in}})$ over $\dalphabet$, and $\aaa_2 = 
(\{Q_i\}, \ab \{\lt_{(a,s_a)}\}, q_{\text{in}})$ over $\dtralphabet$ be two asynchronous 
automata. We define
the local cascade product of $A_1$ and $A_2$, denoted $A_1 \lc A_2$, to be the asynchronous automaton $
	(\{\ls_i \times Q_i\}, \{\Delta_a\}, (s_{\text{in}}, q_{\text{in}}))$
	over $\dalphabet$,
	where, for $a \in \Sigma$ and $(s_a, q_a) \in S_a \times Q_a$,
	$\Delta_a((s_a, q_a))= (\lt_a(s_a), \lt_{(a,s_a)}(q_a))$.
\end{definition}

The operational working of $A=A_1 \circ_{\ell} A_2$ can be understood in terms of
$A_1$ and $A_2$ using the local asynchronous transducer $\chi_{A_1}\colon  \traces \to \tracesdtr$ (associated
with $A_1$) as follows:
for an input trace $t \in \traces$, the run of $A$ on $t$ ends in global state $(s,q)$
if and only if the run of $A_1$ on $t$ ends in global state $s$ and the run of $A_2$ on $\chi_{A_1}(t)$
ends in global state $q$. 
\begin{figure}[ht]
\centering
\begin{tikzpicture}
\draw (0,0) rectangle (1.2,1.2);
\node at (0.6,0.6) {$A$};

\draw (1.2, 0.3) -- (1.5, 0.3);
\draw[->] (1.5, 0.3) -- (1.5,-0.2);
\node at (1.5,-0.35) {\small $(s,q)$};
\draw[-stealth'] (-0.5, 0.6) -- (0,0.6) node [midway, above] {$t$};

\draw (2.8,0) rectangle (4,1.2);
\node at (3.4,0.6) {$A_1$};

\draw (4, 0.3) -- (4.3, 0.3);
\draw[->] (4.3, 0.3) -- (4.3,-0.2);
\node at (4.3,-0.35) {\small $s$};
\draw[-stealth'] (2.3, 0.6) -- (2.8,0.6) node [midway, above] {$t$};

\draw (5.4,0) rectangle (6.6,1.2);
\node at (6,0.6) {$A_2$};

\draw (6.6, 0.3) -- (6.9, 0.3);
\draw[->] (6.9, 0.3) -- (6.9,-0.2);
\node at (6.9,-0.35) {\small $q$};
\draw[-stealth'] (4, 0.6) -- (5.4,0.6) node [midway, above] {\small $\chi_{A_1}(t)$};

\end{tikzpicture}
\caption{Operational view of local cascade product}%
\label{fig:local-cascade}
\end{figure}
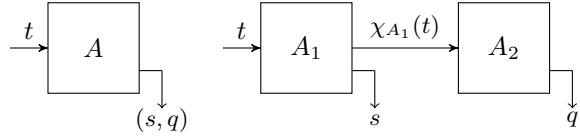
This \emph{operational cascade} of $A_1$ followed by $A_2$ is summarized
in the right part of the Figure~\ref{fig:local-cascade}.
It is not difficult to check that the local cascade product is associative and
$\chi_{A_1 \circ_{\ell} A_2}(t) = \chi_{A_2}( \chi_{A_1}(t))$ for all $t \in \traces$.

Local cascade product is the automata-theoretic counterpart of local wreath product, and Figure~\ref{fig:local-cascade} shows the 
essence of the local wreath product principle discussed in Section~\ref{sec:wpp}. To explain
this further, let $T_1 = (\{ \ls_i \},M)$, $T_2 = ( \{Q_i\},N )$ be two atm's.
Consider an asynchronous morphism $\eta\colon  \traces \to T_1 \wr T_2$. Let $A$ be the asynchronous automaton over 
$\dalphabet$ corresponding to the morphism $\eta$ with a fixed choice of $(s_{\text{in}}, q_{\text{in}}) \in \gs \times Q_\pset$
as the initial global state.

Recall that, as seen in the proof of Theorem~\ref{thm:wpp2}, $\eta$ gives rise to canonical asynchronous morphisms
$\phi\colon  \traces \to T_1$ and $\psi\colon  \tracesdtr \to T_2$ as follows: for $a \in \alphabet$ and $s_a \in S_a$,
\[\eta(a) = (m_a, f_a) \implies \phi(a) = m_a \mbox{~and~} \psi((a,s_a)) = f_a(s_a)\]

Let $A_1$ (resp.\ $A_2$) be the asynchronous automaton over $\dalphabet$ (resp.\ $\dtralphabet$)
corresponding to the morphisms $\phi$ (resp.\ $\psi$) with $s_{\text{in}}$ (resp.\ $q_{\text{in}}$) 
as the initial global state. Then it turns out that $A$ is the local cascade product of
$A_1$ and $A_2$. Following lemma summarizes this. 

\begin{lemma}\label{localcascade}
    Let $\eta \colon \traces \to (\{S_i\}, M)  \wr (\{Q_i\}, N)$ canonically give rise to the two
    asynchronous morphism $\varphi \colon \traces \to (\{S_i\}, M)$ and $\psi \colon \tracesdtr \to (\{Q_i\}, N)$ defined by
    $\eta(a) = (m_a, f_a) \implies \varphi(a) = m_a \text{ and } \psi((a,s_a)) = f_a(s_a)$. If $A$ is the asynchronous
    automata corresponding to $\eta$ with $(s_{\text{in}}, q_{\text{in}})$ as its initial state, and $A_1$ (resp. $A_2$) is the
    asynchronous automata corresponding to $\varphi$ (resp. $\psi$) with $s_{\text{in}}$ (resp. $q_{\text{in}}$) as its initial
    state, then $A = A_1 \lc A_2$.
\end{lemma}
\begin{proof} The proof follows easily from the definitions and skipped.
\end{proof}

\subsection{Global Asynchronous Transducer and its Local Implementation}
Let $\aaa = (\{\ls_i\}, \{\lt_a\}, s_{\text{in}})$ be an asynchronous 
automaton over $\dalphabet$. Recall that the local asynchronous transducer $\chi_A$
preserves the underlying set of events and, at an event, simply records
the previous local states of the processes \emph{participating} in that event.

Now we introduce a natural variant of $\chi_A$ which is called the \emph{global
asynchronous transducer}. In this variant,
at an event, we record the \emph{best global state} that causally precedes the
current event. This is the best global state that the processes participating
in the current event are (collectively) aware of. It is important to note that
the global and local asynchronous transducers coincide in the sequential 
setting.

We first setup some notation.
Based on $\aaa$ and 
$\dalphabet$, we define the alphabet
$\alphabet^{\gs} =  \alphabet \times \gs$
where each letter in $\alphabet$ is extended with global state
 information of $\aaa$. This can naturally be viewed as a distributed
  alphabet $\widetilde{\alphabet^{\gs}}$ where for all $a \in \alphabet$ and $s \in \gs$, we have
  $(a, s) \in \gtralphabet_{i}$ if and only if $a \in \alphabet_i$.

  \begin{definition}[Global Asynchronous Transducer]\label{def:gat}
	  Let $\aaa$ be an asynchronous automaton over $\dalphabet$. The global asynchronous transducer of $\aaa$
 is the map $\theta_\aaa \colon \ab \traces \to \tracesdgtr$ defined as follows.
 If $t = (E, \leq, \labf) \in \traces$, then $\theta_\aaa(t) = (E, \leq, \mu) \in \tracesdgtr$ with the labelling 
$\mu \colon E \to \alphabet \times \gs$ defined by:
\begin{equation*}\forall e \in E, ~\mu(e) = (a,s) \text{ where } a = \labf(e) 
	\text{ and } s = \run_t(\dcset{e}\setminus\{e\})\end{equation*}
\end{definition}
\begin{example} For the same trace $t$ and asynchronous automata $A_\varphi$ from
    Example~\ref{ex:lat}, Figure~\ref{fig:gat-example} shows its global asynchronous
    transducer output $\theta(t)$. Note the difference from Figure~\ref{fig:lat-example}. 
    For example, here the $p_3$-event has process $p_1$ state $2$ in its label
    (which is the best process $p_1$ state in its causal past) even though process
    $p_1$ and process $p_3$ never interact directly.
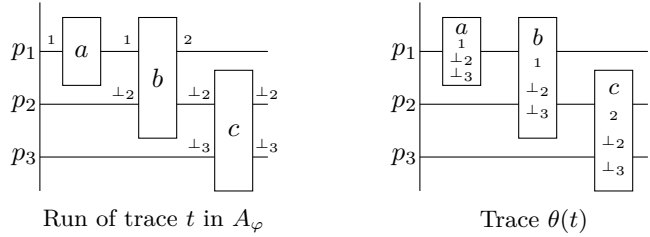
\begin{figure}[ht]
	\centering
	\begin{tikzpicture}
\draw (0,0) -- (0, 5*0.2 + 3*0.5);
%\node at (-0.2, 0.2 + 0.5/2) {$2$};
%\node at (-0.2, 2*0.2 + 0.5 + 0.5/2) {$1$};

\draw (0.3,1.4) rectangle (0.8,2.3);
\draw (1.3,0.7) rectangle (1.8,2.3);
\draw (2.3,0) rectangle (2.8,1.6);

\draw (0,0.2 + 0.5/2) -- (2.3, 0.2 +0.5/2); 
\draw (2.8,0.2 + 0.5/2) -- (3.0,0.2 + 0.5/2); 

\draw (0,2*0.2+0.5+0.5/2) -- (1.3,2*0.2+0.5+0.5/2); 
\draw (1.8,2*0.2+0.5+0.5/2) -- (2.3,2*0.2+0.5+0.5/2); 
\draw (2.8,2*0.2+0.5+0.5/2) -- (3.0,2*0.2+0.5+0.5/2); 

\draw (0, 3*0.2 + 2*0.5 + 0.5/2) -- (0.3, 3*0.2 + 2*0.5 + 0.5/2);
\draw (0.8, 3*0.2 + 2*0.5 + 0.5/2) -- (1.3, 3*0.2 + 2*0.5 + 0.5/2);
\draw (1.8, 3*0.2 + 2*0.5 + 0.5/2) -- (3.0, 3*0.2 + 2*0.5 + 0.5/2);

\node at (0.3 + 0.5/2, 1.6 + 0.5/2) {$a$};
\node at (1.3 + 0.5/2, 0.9 + 0.5 + 0.2/2) {$b$};
\node at (2.3 +  0.5/2, 0.2 + 0.5 + 0.2/2) {$c$};

\node at (-0.2, 0.2 + 0.5/2) {$p_3$};
\node at (-0.2, 2*0.2 + 0.5 + 0.5/2) {$p_2$};
\node at (-0.2, 3*0.2 + 2*0.5 + 0.5/2) {$p_1$};

\node at (0.15, 3*0.2 + 2*0.5 +0.5/2 + 0.15) {\tiny $1$};
\node at (1.15, 3*0.2 + 2*0.5 +0.5/2 + 0.15) {\tiny $1$};
\node at (1.95, 3*0.2 + 2*0.5 +0.5/2 + 0.15) {\tiny $2$};

\node at (1.1, 2*0.2 + 0.5 +0.5/2 + 0.15) {\tiny $\bot_2$};
\node at (2.1, 2*0.2 + 0.5 +0.5/2 + 0.15) {\tiny $\bot_2$};
\node at (2.99, 2*0.2 + 0.5 +0.5/2 + 0.15) {\tiny $\bot_2$};

\node at (2.1, 0.2 + 0.5/2 + 0.15) {\tiny $\bot_3$};
\node at (2.99, 0.2 + 0.5/2 + 0.15) {\tiny $\bot_3$};

\node at (1.5, -0.4) {\small Run of trace $t$ in $A_\varphi$};

\begin{scope}[shift={(5,0)}]
\draw (0,0) -- (0, 5*0.2 + 3*0.5);
%\node at (-0.2, 0.2 + 0.5/2) {$2$};
%\node at (-0.2, 2*0.2 + 0.5 + 0.5/2) {$1$};

\draw (0.3,1.4) rectangle (0.8,2.3);
\draw (1.3,0.7) rectangle (1.8,2.3);
\draw (2.3,0) rectangle (2.8,1.6);

\draw (0,0.2 + 0.5/2) -- (2.3, 0.2 +0.5/2); 
\draw (2.8,0.2 + 0.5/2) -- (3.0,0.2 + 0.5/2); 

\draw (0,2*0.2+0.5+0.5/2) -- (1.3,2*0.2+0.5+0.5/2); 
\draw (1.8,2*0.2+0.5+0.5/2) -- (2.3,2*0.2+0.5+0.5/2); 
\draw (2.8,2*0.2+0.5+0.5/2) -- (3.0,2*0.2+0.5+0.5/2); 

\draw (0, 3*0.2 + 2*0.5 + 0.5/2) -- (0.3, 3*0.2 + 2*0.5 + 0.5/2);
\draw (0.8, 3*0.2 + 2*0.5 + 0.5/2) -- (1.3, 3*0.2 + 2*0.5 + 0.5/2);
\draw (1.8, 3*0.2 + 2*0.5 + 0.5/2) -- (3.0, 3*0.2 + 2*0.5 + 0.5/2);

\node at (0.3 + 0.5/2, 1.6 + 0.5/2 + 0.3) {\small $a$}; % \vspace{1mm} $1$};
\node at (0.3 + 0.5/2, 1.6 + 0.5/2 + 0.1) {\tiny $1$}; % \vspace{1mm} $1$};
\node at (0.3 + 0.5/2, 1.6 + 0.5/2 -0.1) {\tiny $\bot_2$};
\node at (0.3 + 0.5/2, 1.6 + 0.5/2 - 0.3) {\tiny $\bot_3$};

\node at (1.3 + 0.5/2, 1.6 + 0.5/2 +0.2) {\small $b$};
\node at (1.3 + 0.5/2, 0.9 + 0.5 + 0.2/2 +0.2) {\tiny $1$};
\node at (1.3 + 0.5/2, 0.9 + 0.5/2 +0.2) {\tiny $\bot_2$};
\node at (1.3 + 0.5/2, 0.9 + 0.5/2-0.1) {\tiny $\bot_3$};

\node at (2.3 +  0.5/2, 0.2 + 0.5 + 0.2 +0.5/2 +0.2) {\small $c$};
\node at (2.3 +  0.5/2, 0.2 + 0.5 + 0.2/2+0.2) {\tiny $2$};
\node at (2.3 +  0.5/2, 0.2 + 0.5/2 +0.2) {\tiny $\bot_2$};
\node at (2.3 +  0.5/2, 0.2 + 0.5/2 -0.15) {\tiny $\bot_3$};

\node at (-0.2, 0.2 + 0.5/2) {$p_3$};
\node at (-0.2, 2*0.2 + 0.5 + 0.5/2) {$p_2$};
\node at (-0.2, 3*0.2 + 2*0.5 + 0.5/2) {$p_1$};

\node at (1.5, -0.4) {\small Trace $\theta(t)$};
\end{scope}
\end{tikzpicture}
	\caption{Global asynchronous transducer output on a trace}%
	\label{fig:gat-example}
\end{figure}
\end{example}
Below we give a uniform translation from the automaton $A$ to another automaton
$\gossip(A)$ such that the global
asynchronous transducer of $A$ is realized by the local asynchronous transducer of $\gossip(A)$.
\paragraph*{\textbf{Gossip Automaton}} 
It turns out that one must make crucial use of the latest information that the 
agents have about each other when defining the automaton $\gossip(\aaa)$. It has
been shown in \cite{DBLP:journals/dc/MukundS97} that this information can be kept track of by a deterministic 
asynchronous automaton whose size depends only on $\dalphabet$.

To bring out the relevant properties of this automaton, we start with more notation.
Let $t = (E, \leq, \labf) \in \traces$, $c \in \configs{t}$ and $i,j \in \pset$. 
Then $\view{i}{c}$ is the $i$-view of $c$ and it is defined by:
$\view{i}{c} = {\downarrow}(c \cap E_i)$. We note that $\view{i}{c}$ is also a configuration.
It is the “best”
configuration that the agent $i$ is aware of at $c$. It is easy to see that if $\view{i}{c} \neq \emptyset$,
then there exists $e \in E_i$ such that $\view{i}{c} = \dcset{e}$.
For $P \subset \pset$ and $c \in \configs{t}$ , we let
$\view{P}{c}$ denote the set
$\bigcup_{i \in P} \view{i}{c}$. Once again, $\view{P}{c}$ is a configuration. It represents
the collective knowledge of the processes in $P$ about $c$.

For each subset $P$ of processes, the function $\latest_{t,P} \colon 
\configs{t} \times \pset \to P$ is given by $\latest_{t,P}(c,j) = \ell$
if and only if $\ell$ is the least\footnote{we assume an arbitrary total order on $\pset$} member of $P$ with 
$\view{j}{\view{k}{c}} \subseteq ~ \view{j}{\view{\ell}{c}}$
for all $k$ in $P$. In other words, among the agents in $P$, process $\ell$ 
has the best information about $j$ at $c$.

\begin{theorem}[Gossip Automaton~\cite{DBLP:journals/dc/MukundS97}]\label{thm:gossip}	There exists an asynchronous automaton $\gossip = (\{\Upsilon_i\},
	\{\nabla_a\}, v_{\text{in}})$ such that for each $P = \{ i_1, i_2, 
	\ldots, i_n\}$, there exists a function $\mathrm{gossip}_P \colon
	\Upsilon_{i_1} \times \Upsilon_{i_2} \times \ldots \times \Upsilon_{i_n} \times \pset
	\to P$ with the following property. Let $t \in \traces$, $c \in \configs{t}$,
	$j \in \pset$ and let $\run_t$ be the unique run of $\gossip$ over
	$t$ with $\run_t(c) = v$. Then $\latest_{t,P}(c,j) = \mathrm{gossip}_P(v(i_1),
	\ldots, v(i_n),j)$.
\end{theorem}
Henceforth, we refer to $\gossip$ as the \emph{gossip automaton}.

\paragraph*{\textbf{Translation}}
Now we describe the construction of $\gossip(A)$ from $A$. Roughly speaking,
in the automaton $\gossip(A)$, each process $i$ keeps track of the best global
state of $A$ that it is aware of, and its local gossip state in the gossip automaton. 
When processes synchronize, they use the joint gossip-state information to
correctly update the best global state that they are aware of at the synchronizing
event. Of course, they also update their own local gossip states. 

Recall that $\aaa = (\{\ls_i\}, \{\lt_a\}, s_{\text{in}})$. For each $i \in \pset$, let
$\ls_i^\gossip = \Upsilon_i \times \gs$. Further, let $P = \{i_1, i_2, \ldots, i_n\}$. 
We define the function $\globalstate_P \colon \ls_P^\gossip \to \gs$ as follows. Let
$(v_{i_1}, s_{i_1}) \in \ls_{i_1}^\gossip$, $(v_{i_2}, s_{i_2}) \in \ls_{i_2}^\gossip$,
$\ldots$, $(v_{i_n}, s_{i_n}) \in \ls_{i_n}^\gossip$. Then 
\begin{equation*}
	\globalstate_P(((v_{i_1}, s_{i_1}), (v_{i_2}, s_{i_2}),
	\ldots, (v_{i_n}, s_{i_n}) )) = s \in \gs
\end{equation*}
where, for each $i \in \pset$, 
$$s(i) = s_{\ell}(i) \mbox{~with~}  \ell = \mathrm{gossip}_P(v_{i_1}, \ab 
v_{i_2}, \ldots, v_{i_n}, i)$$

We define the asynchronous automaton $\gossip(\aaa)$ to be $(\{\ls_i^\gossip\},\ab 
\{\lt_a^\gossip\}, s_{\text{in}}^\gossip)$. The initial state $s_{\text{in}}^\gossip$
is defined by letting, for each $i \in \pset$, $s_{\text{in}}^\gossip(i) = v_{\text{in}}(i) 
\times s_{\text{in}}$. Now we describe the transition functions $\{\lt_a^\gossip 
\colon \ls_a^\gossip \to \ls_a^\gossip\}_{a \in \alphabet}$. Let $a \in \alphabet$
with $\loc(a) = \{i_1, i_2, \ldots, i_n\}$ and $s_a^{\gossip} \in S_a^{\gossip}$ with
$s_a^{\gossip}(i) = (v_i, s_i)$ for each $i \in \loc(a)$. Suppose $\nabla_a((v_{i_1}, v_{i_2},
\ldots, v_{i_n})) = ((v'_{i_1}, v'_{i_2},\ldots, v'_{i_n}))$. Now we set $\lt_a^{\gossip}
(s_a^{\gossip}) = {s'_a}^{\gossip}$ , such that, for each $i \in \loc(a)$, ${s'_a}^{\gossip}
(i) = (v'_i, s')$ where $s' = \gt_a(s)$ (recall that $\gt_a$ is the global transition function 
of $\aaa$) and $s = \globalstate_{\loc(a)}(s_a^{\gossip})$.

The next proposition says that, the best global state of $A$ that a subset $P$ of processes
are collectively aware of, can be recovered from the {\em local} $P$-joint state of $\gossip(A)$.
Thanks to Theorem~\ref{thm:gossip}, the proof is not very difficult albeit notationally 
somewhat cumbersome.
We skip the proof here.
\begin{proposition}
	With the notation above, the family of functions $\{\globalstate_P\}_{P \subset \pset}$ has the following property. Let 
	$t \in \traces$ and $c \in \configs{t}$. Further, let $\run_t$ and $\run_t^\gossip$ be
	the unique runs of $\aaa$ and $\gossip(\aaa)$ over $t$ with $\run_t^\gossip(c) = s^\gossip$.
	Then $\run_t(\view{P}{c}) = \globalstate_P(s^\gossip_P)$, for each $P \subset \pset$.
\end{proposition}

%Let $\globalstate \colon \cup_{a \in \alphabet}(\{a\} \times S^\gossip_a) \to \alphabet \times \gs$
Let $\globalstate \colon \tralphabetgossip \to \gtralphabet$ be defined as follows.
For $a \in \alphabet$ and $s^\gossip_a \in S^\gossip_a$, we set 
\begin{equation*}\globalstate((a,s^\gossip_a)) =
(a, \globalstate_{\loc(a)}(s^\gossip_a))\end{equation*}

Now we are ready to state one of the main results of this section. It asserts that
the global asynchronous transducer output of $A$ can be obtained from the local
asynchronous transducer output of $\gossip(A)$ by a simple relabelling letter-to-letter morphism
given by the $\globalstate$ function. 
Its proof is immediate and skipped.

\begin{theorem}\label{thm:local-simulates-global}
	Let $A$ be an asynchronous automaton, let $\theta_A \colon \ab \traces \to \tracesdgtr$ be its global asynchronous transducer
	and let $\chi_{\mathcal{G}(A)} \colon \ab \traces \to TR(\wt{\tralphabetgossip})$ be the local asynchronous transducer of
	$\gossip(A)$. Then if $t = (E,\le,\lambda) \in \traces$ and $\chi_{\gossip(A)}(t) = (E,\le,\mu)$, then $\theta_A(t) 
	= (E, \leq, \nu) \in \tracesdgtr$ where, for $e \in E$, $\nu(e) = \globalstate (\mu(e))$
\end{theorem}

\subsection{Global Cascade Product}
Now we are ready to define a cascade model which uses
the global asynchronous transducer.

\begin{definition}[Operational Global Cascade Product]\label{def:gc}
Let $\aaa_1 = (\{\ls_i\}, \{\lt_a\}, 
s_{\text{in}})$ be an asynchronous automaton over $\dalphabet$, and $\aaa_2 = 
(\{Q_i\}, \ab \{\lt_{(a,s)}\}, q_{\text{in}})$ be an asynchronous automaton over 
$\wt{\alphabet^{\gs}}$. Then
	their \emph{operational} global cascade product, denoted by $\aaa_1 \gc \aaa_2$, is a cascade model where, 
for any input trace
$t \in \traces$, $\aaa_1$ runs on $t$ (and `outputs'  $\theta_{\aaa_1}(t)$) and $\aaa_2$ 
runs on $\theta_{\aaa_1}(t)$. See Figure~\ref{fig:global-cascade}.
\end{definition}
\begin{figure}[ht]
\centering
\begin{tikzpicture}%[background rectangle/.style={fill=olive!45}, show background rectangle]
%\draw (0,0) rectangle (1.2,1.2);
%\node at (0.6,0.6) {$A$};

%\draw (1.2, 0.3) -- (1.5, 0.3);
%\draw[->] (1.5, 0.3) -- (1.5,-0.2);
%\node at (1.5,-0.35) {\small $(s,q)$};
%\draw[-stealth'] (-0.5, 0.6) -- (0,0.6) node [midway, above] {$t$};
\draw (2.8,0) rectangle (4,1.2);
\node at (3.4,0.6) {$A_1$};

\draw (4, 0.3) -- (4.3, 0.3);
\draw[->] (4.3, 0.3) -- (4.3,-0.2);
\node at (4.3,-0.35) {\small $s$};
\draw[-stealth'] (2.3, 0.6) -- (2.8,0.6) node [midway, above] {$t$};

\draw (5.4,0) rectangle (6.6,1.2);
\node at (6,0.6) {$A_2$};

\draw (6.6, 0.3) -- (6.9, 0.3);
\draw[->] (6.9, 0.3) -- (6.9,-0.2);
\node at (6.9,-0.35) {\small $q$};
\draw[-stealth'] (4, 0.6) -- (5.4,0.6) node [midway, above] {\small $\theta_{A_1}(t)$};
\end{tikzpicture}%
\caption{Operational view of global cascade product}
\label{fig:global-cascade}
\end{figure}
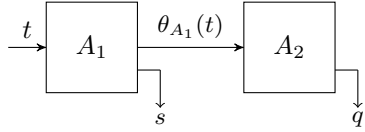
Note that $A_1 \gc A_2$  is not, a priori, an asynchronous automaton, but 
in view of the discussion in the preceeding subsection,
it is simulated by the asynchronous automaton $\gossip(A_1) \circ_{\ell} A_2$. 

For simplicity, we view 
$A_1 \gc A_2$ as an  (asynchronous) 
automaton with same state space as $A_1 \lc A_2$, and extend the notions of run, acceptance
etc.\ to it in a natural way.
Thus, a run of $A_1 \circ_g A_2$ on $t \in \traces$ is a tuple $(\rho_t, \rho_{\theta_{A_1}(t)})$.
As $t$ and $\theta_{A_1}(t)$ have the same set of underlying events with identical causality and concurrent relationships,
both $t$ and $\theta_{A_1}(t)$ admit the same set of configurations. In view of this, we abuse the notation
slightly and write the run as $(\rho^1_t, \rho^2_t)$ where $\rho^1_t\colon  \configs{t} \to S_\pset$
and $\rho^2_t\colon  \configs{t} \to Q_\pset$. Similarly, the label of any event $e$ in $t$ 
	(resp.\ $\theta_{\aaa_1}(t)$) is denoted by $\labf^1(e)$ (resp.\ $\labf^2(e)$). 
	We also use $A_1 \circ_g A_2$ to accept a language by specifying a final subset (of global states) $F \subset \gs \times Q_\pset$.
As expected, an input trace $t = (E, \leq, \lambda) \in \traces$ is accepted if $(\rho^1_t(E), \rho^2_t(E)) \in F$.
Henceforth, we refer to the operational global cascade product as the simply global cascade product.

We consider an asynchronous automaton $A$ 
as a \emph{base} global cascade product.
Now we use structural induction to define the \emph{binary} global cascade product $B_1 \circ_g B_2$ where $B_1$ and $B_2$ are themselves
global cascade products: global states of $B_1 \gc B_2$ is the product of global states of $B_1$ and $B_2$; its run on
$t$ consists of a run of $B_1$ on t and a run of $B_2$ on $\theta_{B_1}(t)$; its global asynchronous transducer is
$\theta_{B_1} \circ \theta_{B_2}$. We also define the notion of a language being accepted by $B_1 \gc B_2$ as expected.

It is easy to see that, for automata $A_1, A_2, A_3$, the global cascade products 
$(A_1 \circ_g A_2) \circ_g A_3$ and $A_1 \circ_g (A_2 \circ_g A_3)$ can be identified naturally in terms of global states, runs, accepted
languages etc. In this sense, the global cascade product is associative. See Figure~\ref{fig:app-global-cascade-associative} for an intuitive explanation of this associativity.

\begin{figure*}[t]
\centering
\begin{subfigure}{\textwidth}
\centering
	\begin{tikzpicture}
\draw (2.6,0) rectangle (3.8,1.2);
\node at (3.2,0.6) {$A_1$};

\draw (3.8, 0.3) -- (4.1, 0.3);
\draw[->] (4.1, 0.3) -- (4.1,-0.2);
\node at (4.1,-0.35) {\small $s'$};
\draw[-stealth'] (2.3, 0.6) -- (2.6,0.6);% node [midway, above] {$t$};

\draw (4.8,0) rectangle (6,1.2);
\node at (5.4,0.6) {$A_2$};

\draw (6, 0.3) -- (6.3, 0.3);
\draw[->] (6.3, 0.3) -- (6.3,-0.2);
\node at (6.3,-0.35) {\small $q'$};
\draw[-stealth'] (3.8, 0.6) -- (4.8,0.6) node [midway, above] {\small $\theta_{A_1}(t)$};

\draw (2.3, -0.6) rectangle (6.6, 1.6);
\draw[-stealth'] (6, 0.6) -- (9,0.6);
\node at (7.8, 0.6) [above] {\small $\theta_{A_1}\circ \theta_{A_2}(t) = t'$};

\draw (9,0) rectangle (10.2,1.2);
\node at (9.6,0.6) {$A_3$};
\draw (10.2, 0.3) -- (10.5,0.3);
\draw[->] (10.5,0.3) -- (10.5, -0.2);
\node at (10.5, -0.35) {\small $r'$};

\draw (6.6, 0.3) -- (7, 0.3);
\draw[->] (7, 0.3) -- (7, -0.2);
\node at (7, -0.35) [xshift=2mm] {\small $(s',q')$};

\draw (1.8, 0.6) -- (2.3, 0.6) node [midway, above] {\small $t$};

\draw[-stealth'] (10.2, 0.6) -- (12,0.6) node [midway, above] {\small $\theta_{A_3}(t')$};
\end{tikzpicture}
\caption{$(A_1 \circ_g A_2) \circ_g A_3$}%
\vspace{1cm}%
\label{fig:global-cascade-associative}
\end{subfigure}
\begin{subfigure}{\textwidth}
	\centering
\begin{tikzpicture}
\draw (12,0) rectangle (13.2,1.2);
\node at (12.6,0.6) {$A_1$};

\draw (13.2, 0.3) -- (13.5, 0.3);
\draw[->] (13.5, 0.3) -- (13.5,-0.2);
\node at (13.5,-0.35) {\small $s''$};
\draw[-stealth'] (11.5, 0.6) -- (12,0.6) node [midway, above] {$t$};

\draw (15.8,0) rectangle (17,1.2);
\node at (16.4,0.6) {$A_2$};

\draw (17, 0.3) -- (17.3, 0.3);
\draw[->] (17.3, 0.3) -- (17.3,-0.2);
\node at (17.3,-0.35) {\small $q''$};
\draw[-stealth'] (17, 0.6) -- (18.4,0.6) node [midway, above]
 {\small $\theta_{A_2}(t'')$};

\draw (18.4,0) rectangle (19.6,1.2);
\node at (19,0.6) {$A_3$};
\draw (19.6, 0.3) -- (19.9,0.3);
\draw[->] (19.9,0.3) -- (19.9, -0.2);
\node at (19.9, -0.35) {\small $r''$};

\draw (15, -0.6) rectangle (20.2, 1.6);
\draw[-stealth'] (13.2, 0.6) -- (15,0.6) node [midway, above] 
{\small $\theta_{A_1}(t) = t''$};
\draw[-stealth'] (15, 0.6) -- (15.8,0.6) node [midway, above] 
{\small $t''$};

%\node at (7.9, 0.6) [above] {$\theta_{A_1 \circ_g A_2}(t)$};
\draw (20.2, 0.3) -- (20.6, 0.3);
\draw[->] (20.6, 0.3) -- (20.6, -0.2);
\node at (20.6, -0.35) [xshift=2mm] {\small $(q'',r'')$};

\draw[-stealth'] (19.6,0.6) -- (22.2, 0.6);
\node at (21.2, 0.6) [above] {$\theta_{A_2} \circ \theta_{A_3}(t'')$};
\end{tikzpicture}
\caption{$A_1 \circ_g (A_2 \circ_g A_3)$}%
\vspace{1cm}%
\label{fig:global-cascade-associative2}
\end{subfigure}
\begin{subfigure}{\textwidth}
	\centering
\begin{tikzpicture}
\draw (2.8,0) rectangle (4,1.2);
\node at (3.4,0.6) {$A_1$};

\draw (4, 0.3) -- (4.3, 0.3);
\draw[->] (4.3, 0.3) -- (4.3,-0.2);
\node at (4.3,-0.35) {\small $s = s' = s''$};
\draw[-stealth'] (2.3, 0.6) -- (2.8,0.6) node [midway, above] {$t$};

\draw (5.4,0) rectangle (6.6,1.2);
\node at (6,0.6) {$A_2$};

\draw (6.6, 0.3) -- (6.9, 0.3);
\draw[->] (6.9, 0.3) -- (6.9,-0.2);
\node at (6.9,-0.35) {\small $q = q' = q''$};
\draw[-stealth'] (4, 0.6) -- (5.4,0.6) node [midway, above] {\small $\theta_{A_1}(t)$};

\draw (8.6,0) rectangle (9.8,1.2);
\node at (9.2,0.6) {$A_3$};
\draw (9.8, 0.3) -- (10.1,0.3);
\draw[->] (10.1,0.3) -- (10.1, -0.2);
\node at (10.1, -0.35) {\small $r = r' = r''$};

\draw[-stealth'] (6.6, 0.6) -- (8.6,0.6) node [midway, above] 
{$\theta_{A_1} \circ \theta_{A_2}(t)$};

\draw[-stealth'] (9.8, 0.6) -- (13,0.6) node [midway, above] 
{$\theta_{A_1} \circ \theta_{A_2} \circ \theta_{A_3}(t)$};
\end{tikzpicture}
\caption{$A_1 \circ_g A_2 \circ_g A_3$}%
\vspace{1cm}%
\label{fig:global-cascade-associative3}
\end{subfigure}
\caption{Associativity of global cascade product}%
\label{fig:app-global-cascade-associative}
\end{figure*}
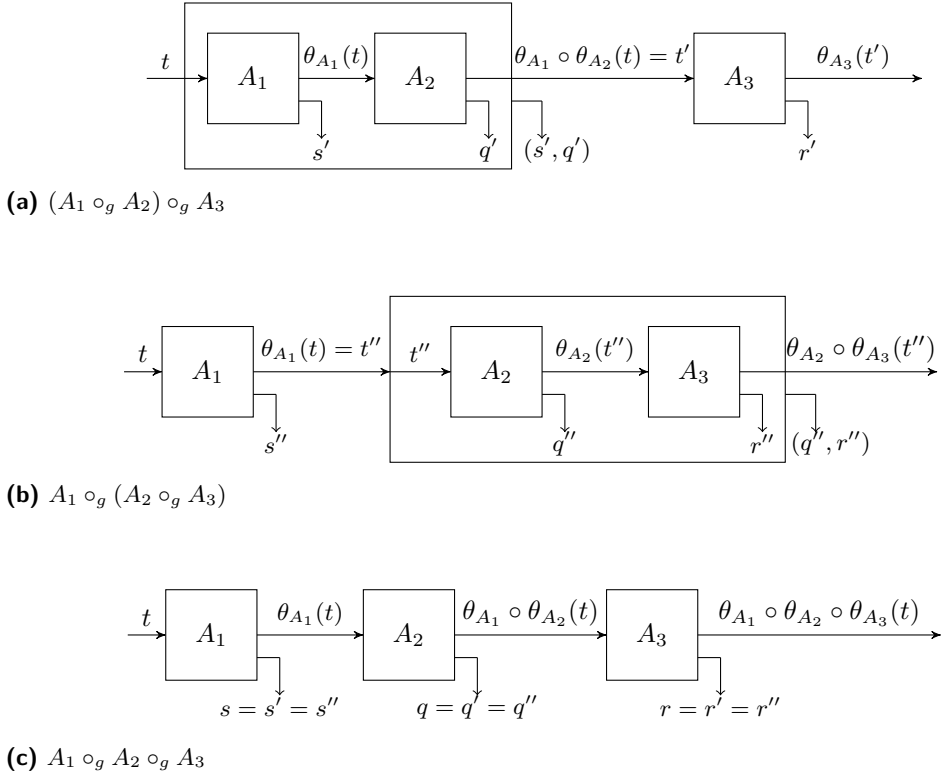

Thanks to this, we can talk about 
the global cascade product of a \emph{sequence} of asynchronous automata. The earlier 
notations, such as $(\rho_t^1, \rho_t^2)$ of a run etc., are easily extended. In this
general setting, the distributed alphabet for the $m$-th asynchronous automata in the
sequence is denoted by $\widetilde{\Sigma^m}$. So, for instance, in Definition~\ref{def:gc}, 
we have $\widetilde{\Sigma^1} = \widetilde{\Sigma}$ and $\widetilde{\Sigma^2} = 
\widetilde{\Sigma^{\gs}}$.

The following \emph{global cascade product principle} is an easy consequence of the definitions.
\begin{theorem}\label{thm:gcwpp2}
	Let $A$ (resp.\ $B$) be a global cascade product over $\dalphabet$ (resp.\ $\dgtralphabet$),
	where $\gs$ is the set of global states of $A$. 
        Then any language $L \subseteq \traces$ accepted by $\aaa \gc \bbb$ 
is a finite union of languages of the form $U \cap \theta_{\aaa}^{-1}(V)$ where $U \subseteq \traces$ is accepted by $\aaa$,
and $V \subseteq \tracesdgtr$ is accepted by $\bbb$. 
\end{theorem}

\section{Temporal Logics, Aperiodic Trace Languages \& Cascade Products}\label{sec:aperiodic}
An automata-theoretic consequence of Theorem~\ref{thm:acyclickr} is that any aperiodic trace language
(that is, a trace language recognized by an aperiodic
monoid) over an acyclic architecture
is accepted by a local cascade product of localized two-state reset automata. 
In this section, we generalize this result to any distributed alphabet, but using global cascade product of the same distributed resets. We call these automata $U_2[\ell]$ as well; that is,
in this section, $U_2[\ell]$ refers to an asynchronous automata whose transition atm (recall discussion of Lemma~\ref{lem:atmza})
is a sub-atm of $U_2[\ell]$ from Example~\ref{ex:u2l}. So $U_2[\ell]=(\{S_i\},\{\lt_a\},
  s_{\text{in}})$ where $S_\ell=\{1,2\}$ and $S_i$ is a singleton set for
  $i\neq\ell$.  A global state of $U_2[\ell]$ is identified with its
  $\ell$-component. Furthermore $\alphabet_\ell$ has two disjoint subsets $R_1, R_2$ that reset the
  state of the automata to $1$ and $2$ respectively. All remaining letters do not change states.

Our proof uses a process-based past local temporal logic (over traces) called
$\loctl$ that exactly defines aperiodic trace languages.
This expressive completeness property of $\loctl$ is an easy consequence of a non-trivial result 
from~\cite{DBLP:journals/iandc/DiekertG06},
where the future version of a similar local temporal logic is shown to coincide with first-order logic 
definable, equivalently, aperiodic trace languages.
The syntax of $\loctl$ is as follows.
\begin{align*}
\text{Event formula }\alpha &= a \mid \neg\alpha \mid \alpha\vee\alpha 	\mid \Y_i{\alpha} 
\mid  \alpha \Si_i \alpha \hspace{9mm}a \in \alphabet, i \in \pset \\
\text{Trace formula }\beta &= \exists_i \alpha \mid \neg \beta \mid \beta \vee \beta 
\end{align*}
The semantics of the logic is given below. 
Each event formula is evaluated at an event
of a trace. Let $t = (E, \leq, \labf) \in \traces$ be a trace with $e \in E$. For any event $x$ in $t$ and $i \in \pset$, we
denote by $x_i$ the unique maximal event of $(\dcset{x} \setminus \{x\}) \cap E_i$, if it exists.
\begin{align*}
t,e & \models a && \text{if } \labf(e) = a \\
t,e & \models \neg \alpha && \text{if } t,e \not\models \alpha \\
t,e & \models \alpha \vee \beta && \text{if } t,e \models \alpha \mbox{~or~} t,e \models \beta \\
t,e & \models \Y_i{\alpha}  && \text{if } e_i \text{ exists, and } t, e_i \models \alpha \\
t,e &\models \alpha \Si_i \alpha' && \text{if } e \in E_i \text{ and }
\exists f \in E_i \text{ such that } f < e \text{ and } t,f \models \alpha'
\\
&&& \text{and } \forall g \in E_i ~~ f < g < e \Rightarrow t,g \models \alpha
\end{align*}
Note that the since operator is a strict version. 
$\loctl$ trace formulas are evaluated for traces, with the following
semantics.
\begin{align*}
t & \models \exists_i \alpha && \text{if there exists a maximal $i$-event $e$ in $t$ such that } t,e \models \alpha
\end{align*}
The semantics of the boolean combinations of trace formulas are obvious. 
Any $\loctl$  trace formula $\beta$ over $\dalphabet$ defines the trace language $L_\beta = \{ t \in \traces \mid t \models \beta \}$. 
The following lemma shows that $\loctl$ is expressively complete.
\begin{lemma}\label{lem:loctlfo}
A trace language is defined by a $\loctl$ formula if and only if it is defined by an
$\mathrm{FO}$ formula.
\end{lemma}
\begin{proof}
    In paper~\cite{DBLP:journals/iandc/DiekertG06}, a process based pure future 
    local temporal logic is shown to be expressively equivalent to first order logic
    over traces. We consider its past dual here with following syntax and semantics 
    for the event formulas.
\begin{align*}
\alpha &= \top \mid a \mid \neg\alpha \mid \alpha\vee\alpha 	
\mid \Y_i{\alpha} 
\mid  \alpha \mathbb{S}_i \alpha \hspace{9mm}a \in \alphabet, i \in \pset
\end{align*}
\begin{align*}
    t,e \models \alpha \mathbb{S}_i \alpha'~~ &\text{ if } \exists f \in E_i
\text{ such that } f \leq e \text{ and } t,f \models \alpha'
\\
 & \text{and } \forall g \in E_i \text{ if } f < g \leq e \text{ then } t,g \models \alpha
\end{align*}
We only need to show that $\loctl$ is expressive enough to define the $\mathbb{S}_i$ operator.
Consider the $\loctl$ formulas $\gamma = \alpha' \vee (\alpha \wedge \alpha \Si_i \alpha')$ and $i = \vee_{a \in \Sigma_i} a$.
Then $\alpha \mathbb{S}_i \alpha'$ is equivalent to the $\loctl$ formula $(i \wedge \gamma) \vee (\neg i \wedge \Y_i \gamma)$.
\end{proof}
\begin{theorem}\label{thm:gcascade}
	A trace language is defined by a $\loctl$ formula if and only if it is accepted by a global cascade product of $U_2[\ell]$.
\end{theorem}
In conjunction with Lemma~\ref{lem:loctlfo}, Theorem~\ref{thm:gcascade} gives a new 
characterization of first order definable trace languages. Before giving its proof,
we provide a local temporal logic characterization for local cascade product of
$U_2[\ell]$.  The corresponding local temporal logic $\sprtl$ is simply the fragment of
$\loctl$ where $\Y_i$ is disallowed.  The semantics is inherited.
It is unknown whether the logic $\sprtl$ is as expressive as $\loctl$.

\begin{theorem}\label{thm:lcascade}%
	A trace language is defined by a $\sprtl$ formula if and only if it is
	recognized by local cascade product of $U_2[\ell]$.
\end{theorem}
\begin{proof}
  $(\Leftarrow)$ Consider a local cascade product $A=U_2[j]\circ_\ell B$.
  By the wreath product principle of Theorem~\ref{thm:wpp2}, and the relation between local cascade of automata and 
  asynchronous morphism into local wreath product of atm, we know that
  any language recognized by $A$ is a union of languages of the form $L_1 \cap
  \chi^{-1}(L_2)$ where $L_1 \subset \traces$ is recognized by $U_2[j]$, the
  language $L_2 \subseteq \tracesdtr$ is recognized by $B$,
  and $\chi$ is the local asynchronous transducer associated to $U_2[j]$ and its initial state, say $1$.

  With global accepting state $2$, the language accepted by $U_2[j]$ is defined
  by the formula $\exists_j(R_2\vee(\neg R_1 \wedge ((\neg R_1)\Si_j R_2)))$.
  The language accepted by $U_2[j]$ with global accepting state 1 can be defined
  with a similar formula.  The difference is that, 1 being the initial state, we
  also have to consider the case where process $j$ contains no events 
  or no events from $R_2$. Hence we use the formula
  $(\neg\exists_j\top) \vee \exists_j
  (R_1 \vee (\neg R_2 \wedge \neg(\neg R_1 \Si_j R_2)))$.
  
  By induction on the number of $U_2[\ell]$s in the local cascade product, we
  know that $L_2$ is $\sprtl$ definable over alphabet $\dtralphabet$.  Thus we
  only need to prove that $\chi^{-1}(L_2)$ is $\sprtl$ definable over
  $\dalphabet$.  We prove this by structural induction on $\sprtl$ formulas over
  $\dtralphabet$.  For $\sprtl$ event formula $\alpha$ over $\dtralphabet$, we
  provide $\hat{\alpha}$ over $\dalphabet$ such that for any trace $t \in
  \traces$, and any event $e$ in $t$, we have $t,e \models \hat{\alpha}$ if and
  only if $\chi(t),e \models \alpha$.  The non-trivial case here is the base
  case of letter formula $\alpha = (a,s_a)$.  If $j \notin \loc(a)$, then
  $\hat{\alpha} = a$, else if ${[s_a]}_j = 2$, then $\hat{\alpha} = a \wedge
  (\neg R_1)\Si_j R_2$.  Other cases can be handled similarly.
	
  \medskip\noindent
	$(\Rightarrow)$ For any $\sprtl$ event formula $\alpha$, we create an
	asynchronous automaton $A_\alpha$ such that for any trace $t$ and its event
	$e$, from the local state ${[\rho_t(\dcset{e})]}_i$ for any $i \in \loc(e)$,
	one can deduce whether $t,e \models \alpha$.
  Furthermore the asynchronous automaton $A_\alpha$ is a local cascade product of $U_2[\ell]$s.
  The construction is done by structural induction on the $\sprtl$ event formulas.

  \subparagraph*{Base Case:} When $\alpha = a \in \alphabet$, let $A_\alpha =
  (\{S_i\}, \{\lt_a\}, s_{\text{in}})$ where $S_i = \{\bot\}$ for all $i \notin
  \loc(a)$, and $S_i = \{\top, \bot\}$ for all $i \in \loc(a)$.  For any
  $P$-state $s$, if for all $i \in P$ we have $s_i = \bot$, then we denote $s =
  \bot$; similarly for $\top$.  Initial state $s_{\text{in}} = \bot$.  For local
  transitions, $\delta_b$ is a reset to $\top$ if $b=a$, and it is a reset to
  $\bot$ otherwise.  By construction we ensure that, for all $i\in\loc(a)$ we
  have ${[\rho_t(\dcset e)]}_i = \top$ if and only if $t,e \models \alpha$.  It
  is also easy to see that $A_\alpha$ is a local cascade product of $U_2[j]$ for
  $j\in\loc(a)$.

	\subparagraph*{Inductive Case:} The non-trivial case is $\alpha = \beta \Si_j
	\gamma$.  By inductive hypothesis, we can assume $A_\beta$ and $A_\gamma$ are
	available.  For simplicity, we assume $A = (\{S_i\}, \{\lt_a\},
	s_{\text{in}})$ simultaneously provides truth value of $\beta$ and $\gamma$ at
	any event.  We construct $B = (\{Q_i\}, \{\lt_{(a,s_a)}\}, q_{\text{in}})$
	over $\dtralphabet$ such that $A\circ_\ell B$ is the required asynchronous 
  automaton.
  Let $Q_i = \{\top, \bot\}$ for all $i \in \pset$. Again, we denote a $P$-state $q$ as
	$\bot$ if $q_i = \bot$ for all $i \in P$; similarly for $\top$. Initial state $q_{\text{in}} = \bot$. For any 
	$a \notin \alphabet_j$, $\lt_{(a,s_a)}$ is a reset to $\bot$. Note that if $\chi$ is the local asynchronous transducer 
	associated with $A$, and in $\chi(t)$ a $j$-event $e$ is 
	labelled $(a,s_a)$, then ${[s_a]}_j$ tells us the truth value of $\beta$ 
	and $\gamma$ at the previous $j$-event $e_j$, if it exists. 
	Let us denote this by ${[s_a]}_j \vdash \beta$ (resp.\ ${[s_a]}_j \vdash \neg\beta$)
	if at the previous
	$j$-event, $\beta$ is true (resp.\ false) according to the $j$ state of $s_a$. Then the transition for $a \in \alphabet_j$ is
	given by
	\begin{align*}
		\lt_{(a,s_a)} &= \text{reset to } \top &&\text{ if } {[s_a]}_j \vdash \gamma \\
		\lt_{(a,s_a)} &= \text{reset to } \bot &&\text{ if } {[s_a]}_j \vdash \neg\gamma 
										\text{ and }{[s_a]}_j \vdash \neg\beta \\
		\lt_{(a,s_a)}(q_a) &= \top &&\text{ if } {[s_a]}_j \vdash \neg\gamma \text{ and } 
									{[s_a]}_j \vdash \beta \text{ and } {[q_a]_j} = \top \\
		\lt_{(a,s_a)}(q_a) &= \bot &&\text{ if } {[s_a]}_j \vdash \neg\gamma \text{ and } 
									{[s_a]}_j \vdash \beta \text{ and } {[q_a]_j} = \bot 
	\end{align*}
	The transitions make sense if we recall the identity 
  $\beta \Si_j \gamma \equiv \Oo_j (\gamma \vee (\beta \wedge (\beta \Si_j 
  \gamma)))$,
  where $\Oo_j \chi \equiv \bot \Si_j \chi$. Note that in the last two transitions above, $\lt_{(a,s_a)}$ is the 
	identity transformation on the process $j$ state; the other processes of $\loc(a)$ can update their states mimicking 
	process $j$ state update if they have the previous process $j$ state information available.
	In view of this, it is easy to verify that $B$ is a local cascade product of
	$U_2[j]$ followed by $U_2[\ell]$ for $\ell \neq j$. 
\end{proof}
We now give the proof of Theorem~\ref{thm:gcascade}.
\begin{proof}[Proof of Theorem~\ref{thm:gcascade}]
  $(\Leftarrow)$ Consider a global cascade product $A = U_2[\ell] \gc B$.
	By the global cascade product principle of Theorem~\ref{thm:gcwpp2}, any
	language recognized by $A$ is a union of languages of the form $L_1 \cap
	\theta^{-1}(L_2)$ where $L_1 \subset \traces$ is recognized by $U_2[\ell]$,
	and the language $L_2 \subseteq \tracesdgtr$ is recognized by $B$, and
	$\theta$ is the global asynchronous transducer associated with $U_2[\ell]$ and
	its initial state, say $1$. We have seen in the proof of Theorem~\ref{thm:lcascade}
	that $L_1$ is $\sprtl$ definable over alphabet $\dalphabet$.
  
%   \textcolor{red}{Since $L_1$ is aperiodic, it is definable in $\loctl$
%   definable over alphabet $\dalphabet$.} 
  
  By induction on the number of $U_2[j]$s
  in the global cascade product, we know that $L_2$ is $\loctl$ definable over
  alphabet $\dgtralphabet$.  Thus we only need to prove that $\theta^{-1}(L_2)$
  is $\loctl$ definable over $\dalphabet$.  We prove this by structural
  induction on $\loctl$ formulas over $\dgtralphabet$.  For $\loctl$ event
  formula $\alpha$ over $\dgtralphabet$, we provide $\hat{\alpha}$ over
  $\dalphabet$ such that for any trace $t \in \traces$, and any event $e$ in
  $t$, we have $t,e \models \hat{\alpha}$ if and only if $\theta(t),e \models
  \alpha$.  The non-trivial case here is the base case of letter formula, say
  $\alpha = (a,2)$.  In this case $\hat{\alpha} = a \wedge \Y_\ell(R_2\vee(\neg
  R_1 \wedge ((\neg R_1)\Si_\ell R_2)))$.  Inductive cases are trivial. For instance 
  $\widehat{\Y_i\alpha}=\Y_i\hat{\alpha}$. Other cases are similar.

  \medskip\noindent
  $(\Rightarrow)$ For any $\loctl$ event formula $\alpha$, we create an
  asynchronous automaton $A_\alpha$ such that for any trace $t$ and its event
  $e$, from the local state ${[\rho_t(\dcset{e})]}_i$ for any $i \in \loc(e)$, one can deduce
  whether $t,e \models \alpha$. Furthermore $A_\alpha$ is a global cascade
  product of $U_2[\ell]$s.  The construction is done by structural induction on
  the $\loctl$ event formulas. Since the $\sprtl$ proof is done,
  we only need to deal with the inductive case of $\Y_j \beta$.
  
  \subparagraph*{Inductive Case:} 
	Suppose $\alpha = \Y_j \beta$. By inductive hypothesis, we can assume
	$A_\beta$ is available, and provides truth value of $\beta$ at any event.
	We construct $B = (\{Q_i\}, \{\lt_{(a,s_a)}\}, 
	q_{\text{in}})$ over $\dgtralphabet$ such that $A_\alpha = A_\beta \gc B$.
	Let $Q_i = \{\top, \bot\}$ for all $i \in \pset$. We denote a $P$-state $q$ as
	$\bot$ if $q_i = \bot$ for all $i \in P$; similarly for $\top$. Initial state $q_{\text{in}} = \bot$. Let $\theta$ be the 
	global asynchronous transducer associated with $A_\beta$. For any trace $t \in \traces$, let $e$ be an event in $t$. If the 
	label of $e$ in $\theta(t)$ is $(a,s)$, then note that $s_j$ tells us the truth value of $\beta$ 
	at the event $e_j$, if it exists. Let us denote this by $s_j \vdash \beta$ and $s_j \vdash \neg\beta$ 
	depending on whether $\beta$ is
	respectively true and false at $e_j$, according to $s_j$. The transition rules are simple
	\begin{align*}
	\lt_{(a,s)} &= \text{reset to } \top &&\text{ if } s_j \vdash \beta \\
	\lt_{(a,s)} &= \text{reset to } \bot &&\text{ if } s_j \vdash \neg\beta
	\end{align*}
	This state update works same for all processes, and it is easy to see that $B$ is, in fact, 
	a local cascade product of $U_2[\ell]$s. However, $B$ requires the global state information from
	$A_\beta$, and so there is a global cascade product between $A_\beta$ and $B$.

  This completes the proof.
\end{proof}

Note that if our postulated decomposition (see Question~\ref{qs-kr}) were true,
it would imply that $\sprtl$ is expressively complete, which would be a stronger
temporal logic characterization for aperiodic, or equivalently first-order logic
definable trace languages than what is currently known.  In particular, by
Theorem~\ref{thm:acyclickr}, $\sprtl$ is expressively complete over tree
architecture.  And this holds true for any distributed alphabet where
Question~\ref{qs-kr} admits a positive answer.

\section{Conclusion}\label{sec:conclusion}
We have presented an algebraic framework equipped with 
wreath products and proved a wreath product principle which is 
well suited for the analysis of trace languages. Building
on this framework, we have postulated  a natural decomposition
theorem which has been proved for the case of acyclic architectures. 
This special case already provides an interesting generalization
of the Krohn-Rhodes theorem. It simultaneously proves Zielonka's theorem for
acyclic architectures.

The wreath product operation in the new framework, when viewed in terms of automata, 
manifests itself in the form of a local cascade product of asynchronous automata. We have also
proposed global cascade products of asynchronous automata and applied
them to arrive at a novel decomposition of aperiodic trace languages.
This is a non-trivial and truly concurrent generalization of the 
cascade decomposition of aperiodic word languages using two-state reset automata.

%% Bibliography
\bibliography{single}

%% Appendix
%\clearpage
%\appendix
%\section{Appendix}
%\input{appendix3.tex}
%\clearpage
%\input{appendix4.tex}
%\clearpage
%\input{appendix5.tex}
%\clearpage
%\input{appendix6-draft.tex}

\end{document}